\documentclass[12pt]{iopart}
 
\usepackage[latin2]{inputenc}
\usepackage{t1enc}
\usepackage{amsthm}
\usepackage{amssymb}
\usepackage{graphics}
\usepackage{graphicx}

\usepackage[normalem]{ulem}
\usepackage{color}
\definecolor{blue}{rgb}{0,0,1}
\definecolor{red}{rgb}{1,0,0}

\newtheorem{Thm}{Theorem}
\newtheorem{Cor}[Thm]{Corollary}
\newtheorem{Lem}[Thm]{Lemma}
\theoremstyle{remark}
\newtheorem{Rem}[Thm]{Remark}
\theoremstyle{definition}

\newcommand{\T}{\mathbb{T}}

\newcommand{\R}{\mathbb{R}}
\newcommand{\C}{\mathbb{C}}
\newcommand{\I}{\mathrm{i}}
\newcommand{\D}{\mathrm{d}}

\newcommand{\id}{\mathrm{id}}

\newcommand{\supp}{\mathop\mathrm{supp}}

\DeclareFontFamily{OT1}{rsfs}{} \DeclareFontShape{OT1}{rsfs}{m}{n}{
<-7> rsfs5 <7-10> rsfs7 <10-> rsfs10}{}
\DeclareMathAlphabet{\mycal}{OT1}{rsfs}{m}{n}

\begin{document}

\vspace*{-2cm}
\title[Multipole Expansion in Time Evolution of Non-linear Dynamical Systems]{On 
the Use of Multipole Expansion in Time Evolution \\
of Non-linear Dynamical Systems and \\ Some Surprises Related to Superradiance\footnote{This paper is dedicated to the memory of our friend P\'{e}ter Csizmadia. P\'{e}ter was a physicist, computer expert and one of the best Hungarian mountaineers. He disappeared in China's Sichuan near the Ren Zhong Feng peak of the Himalayas October 23, 2009. P\'eter was one of the founders and the main developer of our code GridRipper.}}
\author{\fbox{P\'eter Csizmadia}$^1$, Andr\'as L\'aszl\'o$^{1,2}$, Istv\'an R\'acz$^1$}
\address{$^1$WIGNER RCP, H-1121 Budapest, Konkoly-Thege Mikl\'os \'ut 29-33., Hungary}
\address{$^2$CERN, CH-1211 Gen\'eve 23, Switzerland}
\ead{laszlo.andras@wigner.mta.hu, racz.istvan@wigner.mta.hu}
\begin{abstract}
A new numerical method is introduced to study the problem of time evolution of generic non-linear dynamical systems in four-dimensional spacetimes. It is assumed that the time level surfaces are foliated by a one-parameter family of codimension two compact surfaces with no boundary and which are conformal to a Riemannian manifold $\mycal{C}$. The method is based on the use of a multipole expansion determined uniquely by the induced metric structure on $\mycal{C}$. The approach is fully spectral---i.e. it avoids pointwise evaluations of the basic variables---in the {angular} directions. Instead, Gaunt coefficients as matrix elements are used to evaluate multilinear expressions. The dynamics in the complementary 1+1 Lorentzian spacetime is followed by making use of a fourth order finite differencing scheme. In handling the pertinent 1+1 transverse degrees  of freedom the techniques of adaptive mesh refinement (AMR) is also applied.  

In checking the reliability and effectiveness of the introduced new method the 
evolution of a massless scalar field on a fixed Kerr spacetime is investigated. 
In particular, the angular distribution of the evolving field in superradiant 
scattering is studied. The primary aim was to check the validity of some of 
the recent arguments claiming that the Penrose process, or its field 
theoretical correspondence---superradiance---does play crucial role in jet 
formation in black hole spacetimes  while matter accretes onto the central 
object. Our findings appear to be on contrary to these claims as the angular 
dependence of superradiant scattering of massless scalar fields does not show 
any preference of the axis of rotation. In addition, the characteristic 
properties of superradiance in case of a massless scalar field was also 
investigated. On contrary to the general expectations we found that by an 
incident wave packet, which had been tuned to be maximally superradiant, 
the acquired extra energy in the scattering process must be less then $0.1\%$ of the energy sent in.  
It was found that instead of the occurrence of anticipated scale of energy extraction from black 
hole the to be superradiant part of the incident wave packet fail to reach 
the ergoregion rather it suffers a nearly perfect reflection which appears to be 
an interesting phenomenon.

\end{abstract}
\vspace{2pc}

\maketitle

\section{Introduction}
\label{introduction}
\renewcommand{\theequation}{\ref{introduction}.\arabic{equation}}
\setcounter{equation}{0}
\renewcommand{\thefigure}{\ref{introduction}.\arabic{figure}}
\setcounter{figure}{0}
\renewcommand{\thetable}{\ref{introduction}.\arabic{table}}
\setcounter{table}{0}

The basic equations of various theories are non-linear. In studying these types of dynamical systems analytic methods by themselves do not provide a completely satisfactory framework. Therefore it seems to be of fundamental importance to develop numerical methods that are capable to simulate long time evolution of non-linear dynamical systems. Motivated by this sort of necessities in general relativity various groups developed their codes aiming to make progress in the study of astrophysical systems containing black holes and neutron stars. Fully general relativistic
simulations of coalescing  binaries consisting of neutron stars and/or black holes are now possible by making use of variants of the generalized harmonic formulation \cite{pret} and moving puncture approach \cite{punct1,punct2} in the Baumgarte-Shapiro-Shibata-Nakamura (BSSN) formalism \cite{BSSN1,BSSN2} (for a comprehensive review on the recent developments with additional references see, e.g., \cite{Tool,LVR}).

Besides these main stream efforts there are some apparently less ambitious ongoing projects trying to provide precise long term time evolution of various matter fields on fixed stationary background spacetimes which may or may not contain a black hole. In general, these investigations---due to the relative simplicity of the underlying physical system---provide an arena to test some of the new numerical methods before applying them to investigate the aforementioned much more complicated astrophysical systems. Immediate examples for these type of investigations with a stationary black hole as a background spacetime can be found, e.g.\,in \cite{burko} where high order finite differencing was applied, or in \cite{tiglio} where investigations of axial symmetric systems and the use of pseudospectral method (although only moderate angular momentum quantum numbers were involved), or in \cite{scheel} where results on the use of pseudospectral method without assuming axial symmetry but with utilizing parallel computing were reported.
Similar dynamical systems were investigated in a {series} of papers \cite{ZT,Anil,Anil2,OR,RT}. In these papers the viability of the simultaneous use of the techniques of conformal compactification, along with the use of hyperboloidal initial value problem, in numerical simulations were demonstrated. 
 
In the present paper we introduce a numerical method to study the problem of time evolution of generic non-linear dynamical system {in} four-dimensional spacetimes. The time level surfaces are assumed to be foliated by a one-parameter family of codimension two surfaces which are conformal to a compact Riemannian manifold $\mycal{C}$ without boundary. The degrees of freedom in directions tangential to $\mycal{C}$---they are referred as angular directions---are treated with spectral representation (multipole expansion whenever $\mycal{C}$ is homeomorph to a two-sphere $\mathbb{S}^2$) which is based on $L^2$ expansions of the basic variables in terms of the eigenfunctions of the Laplace operator on $\mycal{C}$. The fields in the transverse 1+1 dimensional spacetime directions are evolved by making use of the method of lines based on a fourth order finite difference numerical scheme. The pertinent numerical method incorporates the techniques of the adaptive mesh refinement (AMR). All the operations on the basic variables, involving angular degrees of freedom, are done without applying point{-}wise evaluations, i.e.\,the method is fully spectral, not pseudospectral. 

There are various advantages ensured by this method. Firstly, the usual problems related to the coordinate singularities in the involved angular differential operators can be avoided. Secondly, all the operations in the angular directions which are linear in the basic field variables are exact. What is even more significant---by applying the Sobolev embedding theorem based arguments---all the non-linear operations such as pointwise multiplication or division by fields can also be treated within the spectral representation. The mathematical background of the applied new method---it is assumed to be known but seldom, if ever, collected in a systematic self contained way---is also presented in details in the appendices.  

In practice, all the multipole expansion series are truncated at certain finite order. In this respect the applied method is perturbative. Nevertheless, the error introduced by these approximations can be kept to be at a tolerable low level by increasing the number of the involved modes. The residual error is monitored and the precision, efficiency and the viability of the proposed method have been justified to be satisfactory. 

The introduced new method is applied to investigate the time evolution of a massless Klein-Gordon field on a fixed Kerr black hole spacetime. Within this setting, the angular dependence of the outgoing radiation was studied such that the initial data was fine tuned to have the highest possible potential for superradiance. Our investigations were motivated by some recent attempts trying to provide a physical model yielding high energy collimated matter streams (referred frequently as jets) originating from compact astrophysical objects. 

For instance, in \cite{williams} it was asserted that Penrose process involving Compton scattering on electrons and electron-positron pair production in photon-photon scattering can give rise to ejection of highly collimated matter streams along the axis of rotation. 
In \cite{gariel} an alternative support of these expectations was proposed. In particular, the possible existence of a class of timelike geodesics representing the worldline of the escaping particle yielded in the Penrose process---thereby emerging from the ergoregion---and having the axis of rotation of the black hole as an asymptote was examined. On contrary to these expectations in \cite{takami} where the evolution of a dust sphere falling onto a rotating black hole was considered no significant collimation effect had been found. 

In the present paper an analogous field theoretical model will be investigated. Distinguished attention will be paid to the angular dependence of outgoing radiation yielded by a scattering process with using initial data that has been fine tuned to possess the highest possible potential to generate superradiance. 

Let us also mention here that some preliminary studies of the dynamics of massless scalar field has been done by the present authors in \cite{csizmadia} (see also \cite{RT}\footnote{The code applied 
in \cite{RT} to study the long time evolution was developed essentially by reducing the complexity and by adopting the basic ideas of GridRipper \cite{csizmadia,gridripper3} to the investigated specific problem.}). However, neither the initial condition was fine tuned to generate to be superradiant solution, nor a detailed description of the applied numerical methods was given in either of these works. We would also like to emphasize that in parallel to the preparation this paper the corresponding package of GridRipper with the implementation of the system investigated in this paper is made to be available for public use \cite{gridripper3}.  

Let us also recall here that, based on the estimates in 
\cite{teukolsky,pressteukolsky}, superradiance is expected to be 
more significant whenever higher spin fields such as gravitational radiation is involved. Nevertheless, 
as stated above, throughout this paper considerations are restricted to the 
case of complex scalar fields. Similarly, our results concerning jet formation 
assume that the involved matter is modeled by a complex scalar field. 
Thereby our results do not exclude jet formation found in some recent 
astrophysically motivated more complex magnetohydrodynamical simulations 
(see e.g.\,\cite{luciano,lehner,shapiro}).

The paper is organized as follows. In Section \ref{fieldequations} the physical setup including the field equations and coordinate choices are introduced. Section \ref{numericalevolution} to present an outline of the applied numerical method. Section \ref{boundaryconditions} is to discuss some of the delicate issues related to the applied boundary conditions, while in Section \ref{initialdata} the initial data used in our numerical simulations is introduced. The main results are exposed in Section \ref{resultsanddiscussion}, while our concluding remarks are summarized in Section \ref{summary}. The Appendices are to provide a systematic summary of the mathematical background of the applied new method, in particular, presenting all the details making it possible to use the techniques of multipole expansion in treatment of non-linear dynamical systems. 

\section{Field Equations}
\label{fieldequations}
\renewcommand{\theequation}{\ref{fieldequations}.\arabic{equation}}
\setcounter{equation}{0}
\renewcommand{\thefigure}{\ref{fieldequations}.\arabic{figure}}
\setcounter{figure}{0}
\renewcommand{\thetable}{\ref{fieldequations}.\arabic{table}}
\setcounter{table}{0}

As mentioned above in this paper the evolution of a neutral massless scalar field propagating on the domain of outer communication of a fixed stationary Kerr black hole spacetime{ is considered}. Although the code developed (which can be downloaded from {\cite{gridripper3}}) is also capable to evolve a charged and self-interacting scalar field on a Kerr-Newman background in this paper attention will be restricted to the above mentioned simple case. We would like to mention that even this simple dynamical system is complex enough to test the viability and reliability of the proposed new method based on the spectral method. 

To start off let us recall first the Kerr metric given in  Boyer-Lindquist coordinates $t,r,\vartheta,\varphi$ \cite{wald}. 
The part of the spacetime on which our investigations will be carried out is the domain of outer communication that possesses the product structure $\mathbb{R}^2\times\mathbb{S}^2$ and can be covered by Boyer-Lindquist coordinates $t$, $r$, $\vartheta$ and $\varphi$ taking values 
from the intervals $-\infty<t<\infty$, $0<r<\infty$, $0\leq\vartheta\leq\pi$ and $0\leq\varphi\leq2\pi$. The metric $g$ in these coordinates reads as
\begin{eqnarray}
\label{kerrnewman}
&&\hskip-1cm g=
-\frac{\Delta-a^{2}\,\sin^{2}(\vartheta)}{\Sigma}\,\mathrm{d}t\otimes\mathrm{d}t \cr
&&\hskip-1cm\phantom{g=} -\frac{\,a\,(r^2+a^2-\Delta)\,\sin^{2}(\vartheta)}{\Sigma}\,\left(\mathrm{d}t\otimes\mathrm{d}\varphi+\mathrm{d}\varphi\otimes\mathrm{d}t\right) \cr
&&\hskip-1cm\phantom{g=} +\frac{\Sigma}{\Delta}\,\mathrm{d}r\otimes\mathrm{d}r 
+\Sigma\,\mathrm{d}\vartheta\otimes\mathrm{d}\vartheta
+\frac{\Gamma\,\sin^{2}(\vartheta)}{\Sigma}\,\mathrm{d}\varphi\otimes\mathrm{d}\varphi\,,
\end{eqnarray}
where the smooth functions $\Delta$, $\Sigma$ and $\Gamma$ are 
determined by the relations
\begin{eqnarray}
\label{symbols}
&&\Delta=r^{2}+a^{2}
-2\,M\,r, \cr
&&\Sigma=r^{2}+a^{2}\,\cos^{2}(\vartheta), \cr
&&\Gamma=(r^{2}+a^{2})^{2}-a^{2}\,\Delta\,\sin^{2}(\vartheta).
\end{eqnarray}
The symbols $M$ and $a$ denote the mass and the specific angular momentum of the Kerr black hole spacetime. 
The field equation of a complex valued scalar field $\Phi$ can be written as
\begin{equation}
\label{kgeq}
\nabla^{a}\nabla_{a}\,\Phi=0\,,
\end{equation}
which after the conventional first order reduction, for the vector variable $(\Phi,\Phi_{t},\Phi_{r} )^T$, reads as 
\begin{eqnarray}
\label{kgeqcoord}
&&\partial_{t}\Phi=\Phi_{t}, \cr
&&\partial_{t}\Phi_{t}=\frac{1}{\Gamma}\,\Big(\Delta^{2}\,\partial_{r}\Phi_{r}+2\,\Delta\,(r-M)\,\Phi_{r} +\Delta\,\mathbb{L}_{\mathbb{S}^{2}}(\Phi)\cr
&&\phantom{\partial_{t}\Phi_{t}=\frac{1}{\Gamma}\,\Bigl(}-a^{2}\,\partial^2_{\varphi} \Phi-2\,a\,(r^{2}+a^{2}-\Delta)\,\partial_{\varphi}\Phi_{t}\Bigr)\,, \cr
&&\partial_{t}\Phi_{r}=\partial_{r}\Phi_{t}
\end{eqnarray}
where the differential operator
\begin{equation}
\label{s2laplace}
\mathbb{L}_{\mathbb{S}^{2}}=\frac{1}{\sin\vartheta}\partial_{\vartheta}\left[\sin\vartheta\,\partial_{\vartheta}\right]+\frac{1}{\sin^{2}\vartheta}\partial^2_{\varphi}
\end{equation}
is nothing but the Laplace operator on the unit sphere $\mathbb{S}^{2}$ with its canonical Riemann metric. To get rid of the coordinate singularity 
of the radial differential operator for the first multipole component of $\Phi$ at the origin in the Minkowski limit---that can also be applied in other cases whenever an origin is present in the computational domain---the following conventional trick had been applied. Instead of $\Phi$ the variable $\Psi=r\cdot\Phi$ was evolved using the field equation transformed accordingly. 

In the Kerr case with $M>0$ it turned to be rewarding to use instead of the $r$ and $\varphi$ the new ones $r_{*}$ and $\widetilde\varphi$ defined as 
\begin{equation}
r_{*}(r)=r+\frac{1}{2}\left(\frac{\ln(r-r_{+})}{\kappa_{+}}+\frac{\ln(r-r_{-})}{\kappa_{-}}\right)\,,
\end{equation} 
{\begin{equation}
\tilde{\varphi}(r,\varphi)=\varphi+\frac{a}{r_{+}-r_{-}}\ln\left[\frac{r-r_{+}}{r-r_{-}}\right]\,,
\end{equation}
where 
\begin{equation}
\kappa_{\pm}=\frac{1}{2}\frac{r_{\pm}-r_{\mp}}{r_{\pm}^2+a^2}
\end{equation}
is the surface gravity on the outer and inner event horizon located at $r_{\pm}=M\pm\sqrt{M^2-a^2}$, respectively.
By making use of these coordinates close to the event horizon much better resolution could be achieved which is supported by the fact that the null geodesics of minimal impact can be given as $t\pm r_{*}=const$, $\vartheta=const$,  $\tilde{\varphi}=const$. In consequence of the use of these new coordinates during the evaluation the inverse relation $r=r(r_*)$ had to be determined numerically which was done by implementing a simple Newton-Raphson method.

\section{Numerical Evolution}
\label{numericalevolution}
\renewcommand{\theequation}{\ref{numericalevolution}.\arabic{equation}}
\setcounter{equation}{0}
\renewcommand{\thefigure}{\ref{numericalevolution}.\arabic{figure}}
\setcounter{figure}{0}
\renewcommand{\thetable}{\ref{numericalevolution}.\arabic{table}}
\setcounter{table}{0}

This Section is to provide a short outline of the applied numerical methods. As mentioned already the method is based on multipole expansion on each of the topological two-spheres determined by the $t=const$ and $r=const$ level surfaces. Accordingly, instead of evolving the fields $\Phi$, $\Phi_{t}$ and $\Phi_{r_*}$ themselves their multipole components $[{\Phi}]^m_\ell$, $[\Phi_t]^m_\ell$ and $[\Phi_{r_*}]^m_\ell$---which are functions of $t$ and $r_*$ exclusively yielded by $L^2$ expansions of $\Phi$, $\Phi_{t}$ and $\Phi_{r_*}$ with respect to the spherical harmonics--- 
\begin{equation}
\label{sphset}
\big\{Y{}_{\ell}^{m}\big\vert \ell=0,\dots,\infty,\,m=-\ell,\dots,\ell\big\}\,,
\end{equation}
had been evolved. 
A comprehensive presentation of the mathematical background can be find in the Appendices, while the applied code GridRipper with the implementation of the investigated system can be found at \cite{gridripper3}.

\subsection{Evolution in the $t,r$ section}

The evolution equations for the multipole components $[{{\Phi}}]^m_\ell$, $[{\Phi_t}]^m_\ell$ and $[{\Phi_{r_*}}]^m_\ell$ were solved in the $t-r_*$ plane by making use of the 1+1 dimensional C++ based PDE solver of GridRipper described in details in \cite{gridripper,gridripper2,gridripper3}. The numerical algorithm utilized by this code is based on the method of lines in a fourth order Runge-Kutta scheme such that the spatial derivatives were evaluated with a fourth order symmetric finite difference stencil. To guarantee stability---by suppressing high frequency instabilities---a standard fifth order dissipation term, as proposed by Gustafsson et al \cite{Gustetal} was also applied 
in solving the evolution equations for the multipole components $[{{\Phi}}]^m_\ell$, $[{\Phi_t}]^m_\ell$ and $[{\Phi_{r_*}}]^m_\ell$. Note, that the use of this dissipation term {does not affect} the order of accuracy of the applied numerical scheme{.}
 
The 1+1 algorithm of GridRipper makes use, as a built in package, the techniques of adaptive mesh refinement (AMR) as proposed by Berger-Oliger algorithm \cite{bergeroliger} (see also \cite{gridripper2,gridripper}). The use of AMR is based on the idea that a refining of the spacetime mesh has to be done at those locations where the Richardson error
\begin{equation}
\label{richardson}
\frac{\left\Vert f_{\Delta t,\Delta r}(t,r) - f_{2\Delta t,2\Delta r}(t,r)\right\vert}{2\Delta t \left(2^{q}-1\right)}
\end{equation}
exceeds a predefined threshold, where $f_{\Delta t,\Delta r}$ denotes the numerical solution obtained on a spacetime mesh with $\Delta t$ temporal and 
$\Delta r$ spatial finite difference, $q$ is the order of accuracy of the finite difference scheme, and ${\Vert\cdot\vert}$ is a semi-norm. 
As spatial and temporal refinement is performed simultaneously, the value of the Courant factor---i.e., the ratio of the temporal and spatial 
step size---{remains intact}, thereby in principle the stability of the finite difference scheme {is not affected}.

In our simulations $q$ took the value $4$, while the semi-norm $\Vert\cdot\vert$ was chosen to be the $L^{2}$ norm of the multipole expansion of $\partial_{r_{*}}\Psi${ on each two-sphere}. In order to be able to implement a  relative error type quantity in specifying the tolerable error {this} $L^{2}$ norm of $\partial_{r_{*}}\Psi$ was normalized by the {pertinent} $L^{2}$ norm of $\partial_{r_{*}}\Psi$ on the initial slice.

\subsection{Evolution in the angular 
$\vartheta,\tilde{\varphi}$ section}

The remaining angular $\vartheta,\tilde{\varphi}$ directions were handled by a purely spectral method, completely avoiding point evaluation. 
The expansion coefficients $[{{\Phi}}]^m_\ell$, $[{\Phi_t}]^m_\ell$ and $[{\Phi_{r_*}}]^m_\ell$ were stored in a C++ structure implementing an algebra defined by coefficient-wise linear operations and with pointwise multiplication of the basic variables. It is worth to be emphasized that the viable utility of the latter operation is not obvious at all as the multipole series, by construction, are guaranteed to be convergent only in the $L^2$ sense without an immediate support of their convergence in the pointwise sense. Therefore, in the generic case, the multipole expansion coefficients of pointwise products are not expected to be derived from the multipole coefficients of the factors without evaluating them pointwise and applying a subsequent numerical multipole expansion of the yielded product. Clearly, such a complicated approach would be computationally intensive not allowing the use of multipole expansions with sufficiently large $\ell$ values, e.g.\ $\ell\geq 16$, to make the error introduced by truncation to be tolerably small. 
Nevertheless, whenever the basic variables are known to belong to the class of $C^2$ functions there exists a purely spectral approach that makes the evaluation of their pointwise multiplication possible. All of the underlying ideas---which are of fundamental importance in guaranteeing the effectiveness of the proposed new  method---are justified with mathematical rigor in \ref{sobolevembedding} and \ref{gauntcoefficients}. It is worth to be noted that every solution to a field equation involving second derivatives in the strong sense is of differentiability class $C^2$. As the proposed use of the spectral method avoids pointwise evaluations there is a significant reduction in the required computational power in carrying out full 3+1 dimensional simulations. This reduction is also supported by the fact that the Gaunt coefficients, introduced in \ref{gauntcoefficients}, which are necessary in evaluating products of multipole coefficients have to be calculated only once and stored them in the computer memory during the rest of the simulation. 

In practice, whenever pointwise products of $C^2$ fields truncated at $\ell_{1}$ and $\ell_{2}$ multipole order is evaluated the result shall have non-vanishing coefficients up to $\ell_{1}+\ell_{2}$ multipole order. Therefore, as opposed to linear operations, multiplications do not respect any prefixed maximal expansion order. Correspondingly, in the applied numerical approximation the multipole order of products have to be kept to be bounded which was done by truncating at the value $\max(\ell_{1},\ell_{2})$. Note, however, that convergence tests have to be performed by varying the maximal allowed order $\ell_{\mathrm{max}}$ to justify the viability and the accuracy of the proposed numerical scheme.

In evaluating the time derivative of the basic variables another critical non-linear operation has also to be performed. It is the division by a nowhere vanishing { variable}. The associated difficulties can be overcome by tracing back the operation of pointwise division to the operation of pointwise multiplication with the help of Neumann series expansions. This perturbative method as discussed in details in  \ref{divisionbyfield} and \ref{sobolevconstant} as it is another key ingredient of the proposed new method. It is also shown there that the necessary number of iterations in performing this perturbative division method grows only with the logarithm of the required accuracy.

\subsection{Storage and computational requirements}

The storage of the basic variables via their multipole coefficients becomes even more efficient when the variables may be assumed to be of $C^{\infty}$ class in the angular directions as in that case the corresponding sequence of multipole expansion coefficients are guaranteed to decay faster than any polynomial order as it is justified in \ref{tailsumerrorbounds}. The number of non-vanishing multipole coefficients of a variable, truncated 
at maximal order $\ell_{\mathrm{max}}$, is $(\ell_{\mathrm{max}}+1)^{2}$. Therefore, the storage requirement is quadratic in $\ell_{\mathrm{max}}$. For the number 
of the non-zero Gaunt coefficients the approximate formula $0.7\cdot(\ell_{\mathrm{max}})^{4.7}$ can be verified numerically for $\ell_{\mathrm{max}}\geq 8$. In particular, if the considered problem is axially symmetric the number of non-zero multipole coefficients is only $\ell_{\mathrm{max}}+1$, while the number of Gaunt coefficients necessary to evaluate non-linear terms scale as $0.66\cdot(\ell_{\mathrm{max}})^{2.8}$ provided that $\ell_{\mathrm{max}}\geq 8$. 

It is also informative to compare the computational expense 
estimates to that of other widely used methods, 
for instance pseudospectral methods. 
These store field values over coordinate grids, but evaluate spatial/angular
derivatives in the spectral representation. 
In case of applying the most commonly used basis, the Chebyshev 
polynomials and storing field values over the set of Gauss-Lobatto collocation 
points, one faces the problem of regularizing the coordinate singularities of 
the spherical Laplace operator. In order to avoid the associated 
difficulties, it is quite natural to choose more suitable expansion basis, 
namely spherical harmonics---similarly as in our proposed method.
In adopting this strategy, one must do conversion between grid values 
and spectral coefficients.
In order to estimate the cost of such operation, 
one has to take into account that at each mesh point a sum over the indices $\ell$ 
and $m$ has to be performed. For each mesh point in the $\vartheta$ 
angle coordinate a sum over the  index $\ell$, with evaluations for the involved $m$ values, 
consists of $(\ell_{\mathrm{max}}+1)^2$ terms. In addition, the cost of a sum 
over of the $m$ values---which may be evaluated by making use of 
fast Fourier transform---can be seen to go not better than 
$5(2\,\ell_{\mathrm{max}}+1)\log_{2}(2\,\ell_{\mathrm{max}}+1)$. Taking then into account 
that there exists $\ell_{\mathrm{max}}+1$ pieces of mesh points in $\vartheta$, 
as a minimal estimate of the total cost of a (non-approximate) spectral transformation 
in a pseudospectral method we get  
$5(\ell_{\mathrm{max}}+1)^{3}(2\,\ell_{\mathrm{max}}+1)\log_{2}(2\,\ell_{\mathrm{max}}+1)$. 
On the other hand, the evaluation of non-linear terms using matrix 
products in our method scales as $\sim 0.7\cdot(\ell_{\mathrm{max}}^{4.7})$. Therefore, the 
computational costs of these two methods appear to be comparable although in the range 
$8\leq \ell_{\mathrm{max}}\leq 32$ there is about factor of ten
preference on the side of the fully spectral method. Clearly, at the end 
of the simulation we also need to do pointwise evaluations in extracting the 
physical content of the yielded data. However, these evaluations 
need not to be done on each time level surfaces rather only at some specific 
ones. A slight additional advantage is that the involved vast 
number of matrix multiplications in the spectral case can be paralellized in 
a very effective way. Nevertheless, we admit that it is really the physical 
problem which should decide which method fits better.

\section{Boundary Conditions}
\label{boundaryconditions}
\renewcommand{\theequation}{\ref{boundaryconditions}.\arabic{equation}}
\setcounter{equation}{0}
\renewcommand{\thefigure}{\ref{boundaryconditions}.\arabic{figure}}
\setcounter{figure}{0}
\renewcommand{\thetable}{\ref{boundaryconditions}.\arabic{table}}
\setcounter{table}{0}

The proper treatment of the timelike part of the boundaries is of fundamental importance in both analytic and numerical evolutionary problems \cite{HN,HR,dafermos, helmut}. On numerical side it is only a tiny technical part of the problem that the spatial derivatives cannot be evaluated by making use a symmetric stencil close and at the boundary. The major part of the problems originates from the fact that there remains a freedom in specifying free data on the timelike part of the boundary \cite{HN,HR}. 

In numerical approaches one of the most conventional treatment is to use the Sommerfeld outgoing radiation boundary condition. This method is based on the assumption that at the border the transformed variable $\Psi=r\cdot\Phi$ has vanishing derivative along the outgoing radial null geodesics, which determines the value of $\Psi_{t}$, at the boundary, in terms of the values of $\Psi$ and $\Psi_{r}$ there. Note first that{---given the simple form of the outgoing radial null geodesics in the Kerr spacetime  in terms of the $t,r_*$ coordinates \cite{chandrasekhar}---}it is straightforward to implement this boundary condition in our numerical setup based on the spectral method. Nevertheless, in our first test runs the Sommerfeld boundary condition was found to yield unsatisfactory behavior in long term evolutions. One should keep in mind that even in the simple case of a {scalar field} on Minkowski background, in case of a non-spherical field configurations, only a much more sophisticated treatment \cite{sarbach} can provide a proper numerical treatment, and this approach does not generalize---at least not in a straightforward way---to more general background spacetimes such as the Kerr black hole. 
In virtue of this result one {does} not expect the Sommerfeld outgoing radiation boundary condition to work properly. Indeed, it was shown in \cite{dafermos,allen} that even in the simplest possible case of a massless Klein-Gordon field in Schwarzschild spacetime the asymptotic decay rate of the field may significantly be affected by the use of the Sommerfeld boundary condition. 
Our numerical experiments also justified (see Figure~\ref{bcomp} below) that this outgoing radiation condition yields to significant instabilities at the boundary even in the short-term evolution of strongly non-spherically symmetric configurations.

Another obvious idea is to carry out the numerical simulation near the boundary may be the following. Instead of applying any sort of outgoing radiation condition use the fourth order method of lines everywhere---as it is done in the interior---by making use of an asymmetric stencil close and at the boundaries. This simple minded approach also yields instabilities {developing} at the boundaries{, although this occurs much later than in case of the Sommerfeld boundary condition}.  These type of instabilities most likely are consequences of the sum up of the error produced by the asymmetric fourth order stencils close and at the boundaries. While trying to cure this unfavorable behavior we invented the following simple trick. The order of the finite difference scheme was gradually decreased from 4 to 2 then to 1 such that we still had fourth order symmetric scheme at the last but two points, second order symmetric schemes at the last but one points and a first order upstream or downstream at the {right or left} boundaries{, respectively}. The first order asymmetric treatment at the boundaries could also be considered as a simultaneous combination of a linear extrapolation of the field variable next to the boundary point with the application of a second order symmetric differential scheme. Numerical experiments justified that this simple trick in evaluating the spatial derivatives close and at the boundaries---although with the price of a reduction of the numerical convergence rate there---stabilized the time evolution and, more importantly, guaranteed a long term satisfaction of the energy balance relation even for the evolution of non-spherically symmetric field configurations.

Figure~\ref{bcomp} shows a comparison of the effectiveness of the Sommerfeld boundary condition, the simple $O(4)$ boundary condition---with fourth order asymmetric stencil close and at the boundary---and the developed $O(4-2-1)$ boundary condition---where the order of the finite difference scheme was gradually decreased from $4$ to $2$ then to $1$--- as described in details above. In particular, the time evolution of a {rotating} massless scalar field with a solid toroidal support is considered, by using either of these three boundary conditions, on Minkowski background. It can be seen that the energy outflow pattern at the outer boundary,{ located at} $r=63$, remains, in long term evolution, according to our expectations only for the case of $O(4-2-1)$ boundary condition. In the other two cases so much spurious energy flows back into the computational domain through the outer boundary that kills the evolution at $t\sim 26$ for the Sommerfeld boundary condition and at $t\sim 100$ for the $O(4)$ boundary condition.

\begin{figure}[!ht]
\begin{center}
\includegraphics{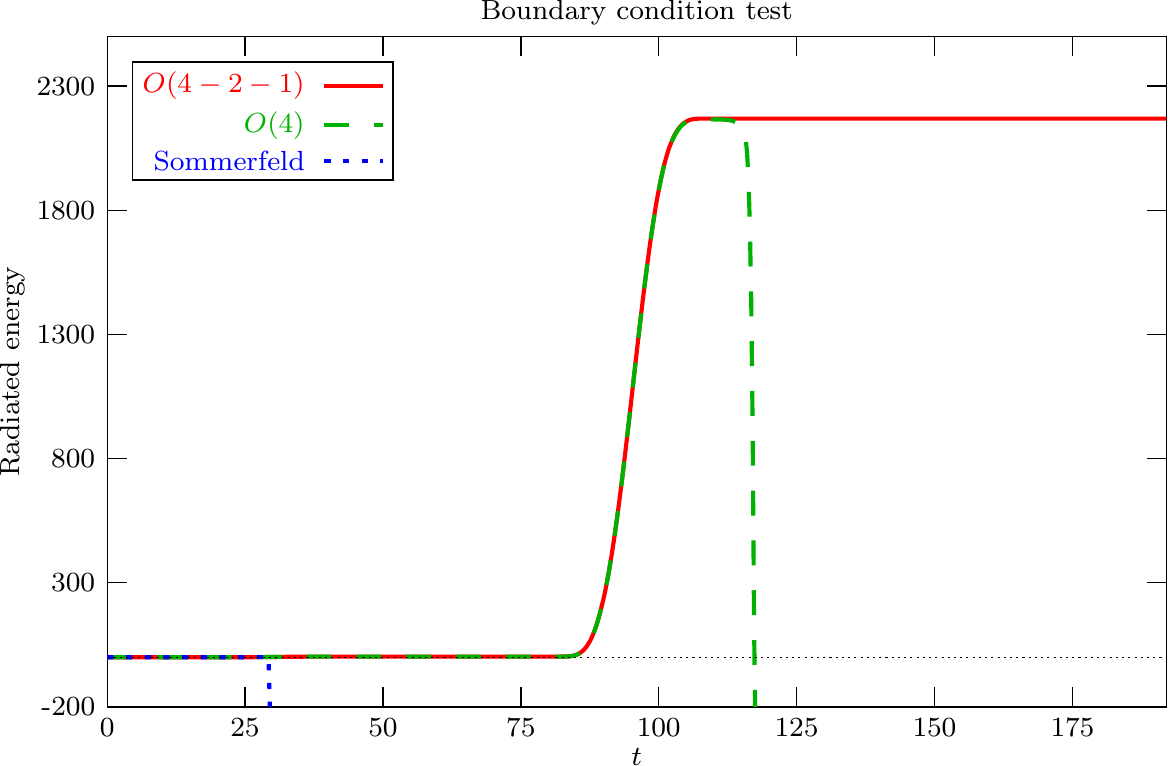}
\vspace*{-0.4cm}
\caption{\label{bcomp}(Color online) 
The energy outflow as a function of temporal coordinate $t$ at the outer boundary, located at $r=63$, for the evolution of an initially rotating massless scalar field on Minkowski background is shown with the application of the Sommerfeld, $O(4)$ and $O(4-2-1)$ boundary conditions, respectively. It can be seen that in case of Sommerfeld or $O(4)$ boundary conditions a spurious energy flow back, from the outer boundary, occurs, while no such spurious energy flow back happens in case of the $O(4-2-1)$ boundary condition. These numerical investigations were carried out using $\ell_{\mathrm{max}}=12$ and $512$ spatial points in the base grid such that $n=5$ AMR refinement levels were allowed.}
\vspace*{-0.4cm}
\end{center}
\end{figure}

It is of obvious interest to know whether the 
$O(4-2-1)$ ``boundary condition'' proposed and used by us is a proper one. 
A sufficiently detailed investigation of this issue exceeds the frame of the 
present paper and the pertinent results will be published elsewhere. Nevertheless, a 
simpleminded explanation concerning the well-posedness of the associated 
initial-boundary-value problem will be given below. Before doing so let us 
point to some of the most fundamental facts which should also support our claim that 
the proposed $O(4-2-1)$ differential scheme does impose proper boundary 
conditions.

\begin{itemize} 
\item[(1)] While by applying the Sommerfeld condition the code 
crashes faster than in case of the $O(4)$ schema where no boundary condition 
at all had been applied. As opposed to this long term stability characterizes 
the use of $O(4-2-1)$ differential scheme.

\item[(2)] The exponential convergence rate, along with the high 
precession of the energy and angular momentum balances 
(see Section~\ref{resultsanddiscussion}), could not be achieved without 
propagating all the physical modes towards the boundary such that they leave 
the computational domain without reflection, i.e., no spurious incoming modes 
are born at the boundaries.
\end{itemize}

The simpleminded argument goes as follows: Assume that we have a 
first order system of hyperbolic field equation of the form $u_t=A u_r+B u$ 
for a vector valued field variable $u$. In applying the $O(4-2-1)$ scheme 
at the last grid point a first order downstream finite difference stencil 
is applied. This, however, can be seen to be equivalent to the application 
of a second order symmetric finite difference stencil combined with a linear 
extrapolation. These two operations guarantee that the second 
$r$-derivative, $u_{rr}$, vanishes there. This, in virtue of this field equation, 
yields a mild restriction---of the type applied in Sections 9-11 of 
\cite{Gustetal}---on the $t$-derivative, $u_t$, at the very last grid point, 
in spite of the fact that apparently only the field equations were imposed there.

\section{Treatment of an origin}
\label{treatmentoforigin}
\renewcommand{\theequation}{\ref{treatmentoforigin}.\arabic{equation}}
\setcounter{equation}{0}
\renewcommand{\thefigure}{\ref{treatmentoforigin}.\arabic{figure}}
\setcounter{figure}{0}
\renewcommand{\thetable}{\ref{treatmentoforigin}.\arabic{table}}
\setcounter{table}{0}

The presence of an origin on the time slices always requires a very careful and precise treatment. An origin shows up in various physically realistic situations. For instance, if the background is the Minkowski spacetime or in case of fully dynamical spacetimes containing a pulsating neutron star. It worth to be emphasized that the method outlined below is applicable not only in case of a Minkowski background but in the generic case of fully dynamical spacetimes, as well. 

Before proceeding note that in the fully dynamical situations the time slices might contain more than one origin or no origin at all (see, e.g.~\cite{CSPIR} for explicit examples). Nevertheless, in this Section attention will be restricted to the conventional single origin case. In order to avoid ambiguities the $r$-coordinate is arbitrary and the origin is assumed to be located at $r=0$.

In order to avoid the appearance of the usual $r$-coordinate singularity in evaluating the radial part of the Laplace operator at the origin, the new basic variable $\Psi=r\cdot\Phi$ has already been introduced. Whenever an origin is located at $r=0$ a symmetric fourth order stencil can be applied in determining the spatial derivatives at and in a neighborhood of the origin based on the observation that the multipole coefficients $[\Psi]_{\ell}^{m}$ may formally be extended to negative radii according to the rule
\begin{equation}
\label{reflection}
[\Psi]_{\ell}^{m}(-r)=(-1)^{\ell+1}[\Psi]_{\ell}^{m}(r)\,.
\end{equation}
This relation follows from the assumption that the original field variable $\Phi$ is at least $C^{1}$---thereby it is $C^{1}$ along arbitrary straight lines through the origin---and from the reflection property of the spherical harmonics $Y{}_{\ell}^{m}$ under the transformation $\vartheta,\varphi \mapsto \pi-\vartheta,\varphi+\pi$.\footnote{
Indeed, the radial derivative $\partial_{r}\Phi$ of $\Phi(r,\vartheta,\varphi)$ 
corresponds to the directional derivative of $\Phi$ along the radial direction 
characterized by the certain constant values of the angles $\vartheta,\varphi$. As $\Phi$ is required to be $C^1$, the 
directional derivative $\partial_{r}\Phi$ may, then, be numerically evaluated by making use of a symmetric stencil of the applied finite difference scheme such that values of $\Phi$ from both sides of the origin are used at and closed to the origin. 
Nevertheless, on the the opposite side the values of  $\Phi$ can be determined as
$
\Phi(r, \pi-\vartheta,\varphi+\pi)=
\sum_{\ell=0}^{\infty}\sum_{m=-\ell}^{\ell}\Phi_{\ell}^{m}(r)\,Y{}_{\ell}^{m}(\pi-\vartheta,\varphi+\pi)=
\sum_{\ell=0}^{\infty}\sum_{m=-\ell}^{\ell}\Phi_{\ell}^{m}(r)\,(-1)^{\ell}Y{}_{\ell}^{m}(\vartheta,\varphi)\,.
$
Thus, by formally extending the functions $\Phi_{\ell}^{m}$ to negative radii---by making use of the rule $\Phi_{\ell}^{m} \mapsto (-1)^{\ell}\Phi_{\ell}^{m}$---the radial derivative $\partial_{r}\Phi_{\ell}^{m}$ of the functions $\Phi_{\ell}^{m}$ can numerically be evaluated.}

In addition, if $\Phi$ is assumed to be $C^{2}$---this assumption should not be considered as extreme especially if one recalls that $\Phi$ is subject to (\ref{kgeq})---its spatial Laplacian 
\begin{equation}
{\frac{1}{r}\partial^2_{r}(r\cdot\Phi)+\frac{1}{r^{2}}\mathbb{L}_{\mathbb{S}^{2}}(\Phi)}
\end{equation}
has to be finite at the origin. From this, along with subsequent applications of l'Hopital rule, the relations
\begin{eqnarray}
\label{lhopital1}
&& \ell=0:   \,\;\partial_{r}[\Phi]_{\ell}^{m}=0, \cr
&& \ell=1:   \,\;[\Phi]_{\ell}^{m}=0, \cr
&& \ell\geq2: \,\;[\Phi]_{\ell}^{m}=0, \,\;\partial_{r}[\Phi]_{\ell}^{m}=0
\end{eqnarray}
follow. This, however, along with the substitution $\Psi=r\cdot\Phi$, implies that the relations 
\begin{eqnarray}
\label{lhopital2}
&&\hskip-.95cm \ell=0:    \,\;[\Psi]_{\ell}^{m}=0,\;\partial^2_{r}[\Psi]_{\ell}^{m}=0, \cr
&&\hskip-.95cm \ell=1:    \,\;[\Psi]_{\ell}^{m}=0, \;\partial_{r}[\Psi]_{\ell}^{m}=0,\cr
&&\hskip-.95cm \ell\geq2: \,\;[\Psi]_{\ell}^{m}=0, \;\partial_{r}[\Psi]_{\ell}^{m}=0, \;\partial^2_{r}[\Psi]_{\ell}^{m}=0
\end{eqnarray}
hold for the multipole components $[\Psi]_{\ell}^{m}$ at the origin.

In consequence of the algebraic relations $[\Psi_{\ell}]^{m}=0$ and $\partial_{r}[\Psi]_{\ell}^{m}=0$, these hold for $\ell\neq0$ at the origin, the field equations read there as
\begin{eqnarray}
\label{origin}
&& \partial_{t}[\Psi]_{\ell}^{m}=0, \cr
&& \partial_{t}[\Psi_{t}]_{\ell}^{m}=0, \cr
&& \partial_{t}[\Psi_{r}]_{\ell}^{m}=\left\{ \begin{array} {r l}
 \partial_{r}[\Psi_{t}]_{\ell}^{m} , & 
{\rm if}\, \ell=0 \cr 0 , & {\rm otherwise}, \end{array}
\right.
\end{eqnarray}
which are completely regular at $r=0$. In spite of this apparently straightforward regularization of the singular terms at the origin a simpleminded numerical implementation of (\ref{origin}) still yield unstable evolutions for non-spherically symmetric configurations. A close look at the evolution justifies that numerical error starts to grow very quickly at the grid-point next to the origin. This can be understood by recalling that in the evaluation of the `$0/0$' type term 
\begin{equation}
\label{originterm} 
{\frac{1}{r^{2}}\mathbb{L}_{\mathbb{S}^{2}}(\Psi)}
\end{equation}
the higher multipole components acquire larger weight, which significantly magnifies the related numerical error. Nevertheless, this difficulty can also be overcome by a systematic application of the algebraic relations formulated by (\ref{reflection}) and (\ref{lhopital2}). By requiring these conditions to hold---where the first and second derivatives are assumed to be evaluated as dictated by the applied fourth order symmetric finite difference scheme---it turns out that the values of $[\Psi]_{\ell}^{m}$ at the origin and next to the origin are algebraically determined by the values of $[\Psi]_{\ell}^{m}$ next to next to the origin, with the only exception $[\Psi_{r}]_{0}^{0}$ which evolves according to the regular field equation $\partial_{t}[\Psi_{r}]_{0}^{0}=\partial_{r}[\Psi_{t}]_{0}^{0}$ [see (\ref{origin})] at $r=0$.

\section{Initial Data and the Applied Grid}
\label{initialdata}
\renewcommand{\theequation}{\ref{initialdata}.\arabic{equation}}
\setcounter{equation}{0}
\renewcommand{\thefigure}{\ref{initialdata}.\arabic{figure}}
\setcounter{figure}{0}
\renewcommand{\thetable}{\ref{initialdata}.\arabic{table}}
\setcounter{table}{0}

A generic initial data specification to our evolution equations (\ref{kgeqcoord}) is composed by three functions $\phi$, $\phi_{t}$ and $\phi_{r}$ specified on the $t=0$ initial data hypersurface, denoted by $\Sigma_0$, such that beside the trivial constraint $\phi_{r}=\partial_{r}\phi$ for the corresponding solution $\Phi$ the relations $\Phi|_{\Sigma_0}=\phi$ and $\Phi_{t}|_{\Sigma_0}=\phi_{t}$ also hold. It is straightforward to recast such an initial data specification for the rescaled field variable $\Psi=r\cdot\Phi$ which is given as a function of the coordinates $t,r_{*},\vartheta,\tilde{\varphi}$ defined in Section \ref{fieldequations}.

\subsection{Superradiance}

Before proceeding and introducing our choice for the only freely specifiable functions $\psi$ and $\psi_t$ on $\Sigma_0$ let us recall some simple facts related to superradiance. First of all, as it was shown first by Carter in \cite{carter} in the coordinates $t,r_{*},\vartheta,{\varphi}$ the d'Alembert operator separates for the $t$-Fourier transformed field. More precisely, the temporal Fourier transform, ${\mycal{F}}\Phi$, of a solution $\Phi$ to (\ref{kgeqcoord}) may be decomposed as 
\begin{equation}
\label{kerrrepresentation}
{\mycal{F}}\Phi(\omega,r_{*},\vartheta,{\varphi})=\frac{1}{\sqrt{r^{2}+a^{2}}}\sum_{\ell=0}^{\infty}\sum_{m=-\ell}^{\ell}R_{\ell,\omega}^{m}(r_{*})S_{\ell,a\omega}^{m}(\vartheta ,{\varphi})\,,
\end{equation}
where $\omega$ is the frequency in the time translation direction and 
$S_{\ell,a\omega}^{m}$ denotes the oblate spheroidal harmonic function with oblateness parameter $a\omega$ and with angular momentum quantum numbers $\ell,m$---they are eigenfunctions of a self-adjoint operator---, while for the radial functions $R_{\ell,\omega}^{m}$ a one-dimensional Schr\"odinger equation of the form 
\begin{equation}\label{schro}
\frac{d^2{R}_{\ell,\omega}^{m}}{dr_{*}^2}+\left[\left(\omega-\frac{m a}{r^2+a^2}\right)^2+\Delta\cdot V_{\ell,\omega}^{m}(r_{*})\right]\,{R}_{\ell,\omega}^{m}=0\,,
\end{equation}
with suitable real potentials 
$V_{\ell,\omega}^{m}(r_{*})$ can be derived from the field equation (\ref{kgeqcoord}).

The conventional argument ending up with the phenomenon called superradiance goes as follows. 
Physical solutions to (\ref{schro}) are supposed to possess the asymptotic behavior 
\begin{equation}\label{asympt}
{R}_{\ell,\omega}^{m}\sim 
\cases{
e^{-i\omega r_{*}} + \mathcal{R} \,e^{+i\omega r_{*}}  \;\;\mathrm{ as }\;\;r\rightarrow \infty
\cr
\mathcal{T}\,e^{-i(\omega- m \Omega_{H})r_{*}}         \;\;\;\;\;\mathrm{ as }\;\;r\rightarrow r_+
}
\end{equation} 
where 
$\Omega_{H}$ denotes the angular velocity of the black hole with respect to the asymptotically stationary observers \cite{wald}, and with reflection and transmission coefficients, $\mathcal{R}$ and $\mathcal{T}$ \cite{laguna}, respectively. Notice that this asymptotic behavior (\ref{asympt}) presumes the existence of a transmitted wave submerging into the ergoregion. By evaluating the Wronskian of the corresponding fundamental solutions, 
``close'' to infinity and ``close'' to the horizon, it can be shown that the reflection and transmission coefficients satisfy the relation 
$(\omega- m \Omega_{H})\,\left\vert\mathcal{T}\right\vert^{2} =(1-\left\vert\mathcal{R}\right\vert^{2})\,\omega$ \cite{yau}. Thereby, whenever 
$\left\vert\mathcal{R}\right\vert>1$---or equivalently, whenever $\vert\mathcal{T}\vert$ does not vanish and the inequality $0<\omega<m\Omega_{H}$ holds---positive energy is supposed to be acquired by the backscattered scalar wave due to its interaction with the Kerr black hole in the ergoregion. 
This phenomenon is referred as superradiant scattering which is also known as the field theoretical correspondence of the Penrose process derived in context of point particle mechanics \cite{penrose}. 

\subsection{Initial data}\label{supercheck}

In applying the introduced new numerical method our primary interest was to study the angular dependence of superradiant scattering. The applied initial data was specified accordingly---by applying an approach analogous to that of \cite{krivan,laguna}---{and} it was fine tuned to maximize {the effect of superradiance}.

However, to investigate a clear manifestation of superradiance in a fully dynamical process---i.e.\ the way an incident scalar wave acquires extra energy by submerging into the ergoregion and then carrying it away from the central region---the initial data we applied is of compact support such that it is separated from the ergoregion on the initial time slice. Thereby, the initial data we have applied differs significantly, in its fundamental character, from that of \cite{krivan,laguna}. To fulfill the above mentioned requirements the initial data for the rescaled field variable $\Psi=r\cdot\Phi$ was chosen as 
\begin{eqnarray} 
\label{initdata}
&&\hskip-.5cm\psi(r_{*},\vartheta,\tilde{\varphi})=e^{-\I\omega_{0}(r_{*}-r_{*0})}f(r_{*}-r_{*0})\,Y{}_{\ell}^{m}(\vartheta,\tilde{\varphi})\,,\cr
&&\hskip-.5cm\psi_{t}(r_{*},\vartheta,\tilde{\varphi})=-\I\omega_{0}\psi(r_{*},\vartheta,\tilde{\varphi})+e^{-\I\omega_{0}(r_{*}-r_{*0})}f'(r_{*}-r_{*0})\,Y{}_{\ell}^{m}(\vartheta,\tilde{\varphi})\,,\cr
&&\hskip-.5cm\psi_{r_{*}}(r_{*},\vartheta,\tilde{\varphi})=\partial_{r_{*}}\psi(r_{*},\vartheta,\tilde{\varphi})\,,
\end{eqnarray} 
where $f:\R\rightarrow\C$ is a smooth function of compact support, $f'$ denotes its first derivative and $\omega_{0}$, $r_{*0}$ are some real parameters.
Note that the appearance of the extra $r$ factor in $\Psi$ has no effect on the above recalled argument concerning the appearance of superradiance. Indeed, this factor may be suppressed by redefining the function $f$ that has not been specified yet.

It can be seen that in an asymptotic region of the Kerr background (or everywhere if the background is the Minkowski spacetime) the initial data specification (\ref{initdata}) yields an inward traveling spherical wave packet starting with a radial profile $f(r_{*}-r_{*0})$. Accordingly, in an asymptotic region the solution in a sufficiently small neighborhood of the initial data surface might be approximated as  
\begin{eqnarray}
\label{minkowskisolution}
\Psi(t,r_{*},\vartheta,\tilde{\varphi}) \approx e^{-\I\omega_{0}(r_{*}-r_{*0}+t)} f(r_{*}-r_{*0}+t)\,Y{}_{\ell}^{m}(\vartheta,\tilde{\varphi}).
\end{eqnarray}
It is informative to have a look at the temporal Fourier transform, ${\mycal{F}}\Psi$, of this approximate solution $\Psi$ that reads as 
\begin{eqnarray}
\label{fminkowskisolution}
{\mycal{F}}\Psi(\omega,r_{*},\vartheta,\tilde{\varphi}) \approx e^{-\I\omega(r_{*}-r_{*0})} {\mycal{F}}f(\omega-\omega_{0})\,Y{}_{\ell}^{m}(\vartheta,\tilde{\varphi}),
\end{eqnarray}
where $\omega$ being the temporal frequency while ${\mycal{F}}f$ stands for the Fourier-transform of $f$. Assuming that ${\mycal{F}}f$---playing the role of a frequency profile function---is sufficiently narrow the approximate solution (\ref{minkowskisolution}) looks almost like a monochromatic spherical wave solution similar in nature to the ingoing part of the wave determined by relations (\ref{kerrrepresentation}) and (\ref{asympt}). 
Accordingly,  by tuning $\omega_{0}$ such that the energy flux absorbed by the black hole to become negative---this is expected to be achieved by choosing $\omega_{0}$ such that $0<\omega_{0}<m\Omega_{H}$---one would expect that a to be superradiant solution is yielded. It can also be seen that the energy extraction may be maximized by choosing $\omega_{0}=\frac{1}{2}m\Omega_{H}$ and, in addition, by guaranteeing that $\int_{0}^{m\Omega_{H}}\vert {\mycal{F}}f\vert^2(\omega-\omega_{0})\,\mathrm{d}\omega\approx\int_{-\infty}^{\infty}\vert {\mycal{F}}f\vert^2(\omega-\omega_{0})\,\mathrm{d}\omega$, which happens whenever the frequency spectrum is narrow enough to be entirely included by the superradiant frequency regime. 

We would like to emphasize that the above outlined construction of a to be superradiant initial data specification involves a number of heuristics assumptions. For instance, the Fourier spectrum (\ref{fminkowskisolution}) is assumed to represent the Fourier transform of the purely inward traveling wave (\ref{minkowskisolution}) and whence the contribution from back scattering is completely neglected. 
To convince ourselves, in investigating the time evolutions of specific initial data choices, the power spectrum in temporal frequency of a supposed to be superradiant solution was also determined at a constant $r_*$ sphere which is located towards the black hole with respect to the compact support of the initial data. As it can be seen on Figure~\ref{spectrum1d} the solution remains in the desired frequency regime.

Based on the above outlined reasoning in our numerical simulations the radial profile function $f:\R\rightarrow\C$ was chosen to possess the form
\begin{eqnarray}
\label{profilefunction}
\hskip-0.3cm{f_{w}(x)=
\left\{
\begin{array} {r l}
e^{\left[-\left\vert\frac{w}{x+\frac{w}{2}}\right\vert-\left\vert\frac{w}{x-\frac{w}{2}}\right\vert+4\right]}\,, & \mathrm{if } x\in[-\frac{w}{2},\frac{w}{2}]\cr
 0\,, & \mathrm{otherwise} \,,
\end{array}
\right.}
\end{eqnarray}
which is a smooth function of the real variable $x$ with compact support $[-\frac{w}{2},\frac{w}{2}]$. To be compatible with our most important physical requirements that yields an incident wave packet that may acquire extra energy after penetrating through the ergoregion the initial parameter $r_{0*}$ in (\ref{initdata}) was chosen to be sufficiently large to have a clear separation of the support of the initial data and the ergoregion on $\Sigma_0$. 

\subsection{Grid size and parameters}

The radial extent of the computational domain used in our simulations was chosen to be the closed interval $-64\leq r_{*}\leq64$ in the massive case ($M=1$), whereas the closed interval $0\leq r_{*}\leq64$ in case of the Minkowski limit ($M=0$). The specific angular momentum parameter $a$ of the Kerr background was always chosen to be $0.9$ while the Schwarzschild limit was achieved by taking $a=0$. The evolution of the system was investigated in the time interval $0\leq t\leq 192$. The fine tuned parameters of the initial data (\ref{initdata})---tuned to have the largest possible effect in superradiance---were $\omega_{0}=0.313394503136629$, $r_{*0}=31.8229346475152$, $w=35.3679317843828$, while the angular and azimuthal mode numbers $\ell$ and $m$ were fixed by choosing $\ell=2$ and $m=-2,0,2$. Accordingly, the initial data had pure quadrupole character, while for $m$ the values $-2$, $0$ and $2$ signifies counter-rotating, non-rotating and co-rotating initial distributions, respectively. In virtue of the above discussion we may only expect the appearance of superradiance in the co-rotating case with $m=2$, while no or negligible effect may be anticipated in the non-rotating or counter-rotating cases with $m=0$ or $m=-2$, respectively. 

In order to justify the above very specific choice made for the parameters $\omega_{0}$, $r_{*0}$ and $w$ let us recall the list of requirements they have to satisfy. 
\begin{itemize}
\item $(w/4)^{-1}\ll m\Omega_{H}$ \  $\rightleftharpoons$ \   The width of frequency profile should be much smaller than the width of the superradiant frequency domain.
\item $w\ll r_{*,\mathrm{max}}-r_{*,\mathrm{ergosphere}}$ \  $\rightleftharpoons$ \  The width of the initial data has to be much smaller than the part of the domain of outer communication outside to the ergosphere and covered by the grid.
\item $w\gg \Delta{r_{*}}$ $\rightleftharpoons$ The width and ramp of the wave packet has to be much larger than the spatial resolution of the base grid applied in AMR.
\item $\omega_{0}\ll \Delta{t}^{-1}$\  $\rightleftharpoons$  \ The leading frequency of the initial data has to be much smaller than the maximal frequency allowed by the temporal resolution of the base grid.
\item $r_{*0}-\frac{1}{2}w>r_{*,\mathrm{ergosphere}}$\  $\rightleftharpoons$  \ The support of the initial data has to be well separated from the ergoregion.
\item $r_{*0}+\frac{1}{2}w<r_{*,\mathrm{max}}$\  $\rightleftharpoons$  \ The support of the initial data has to be included with suitable margins by the radial computational domain.
\end{itemize}

\subsection{Generic initial data for GridRipper}

Let us finally mention that in spite of the fact that in the investigations reported in this paper the initial data is always of pure multipole type in our code GridRipper (that can be downloaded from \cite{gridripper3}) the generic case---whenever a multipole expansion of the initial data is required---is also implemented (see, e.g., \cite{csizmadia} for an application). In the current version of GridRipper this is done by simply integrating numerically the product of the basic variables with $\overline{Y}{}_{\ell}^{m}$ over the $r_{*}=const$ angular spheres on $\Sigma_0$. In order to make this part computationally inexpensive---reducing thereby the required computational time to the order of seconds---the very efficient and precise two dimensional adaptive Genz-Malik (AGM) method \cite{genzmalik} is applied.

\section{Numerical investigations}
\label{resultsanddiscussion}
\renewcommand{\theequation}{\ref{resultsanddiscussion}.\arabic{equation}}
\setcounter{equation}{0}
\renewcommand{\thefigure}{\ref{resultsanddiscussion}.\arabic{figure}}
\setcounter{figure}{0}
\renewcommand{\thetable}{\ref{resultsanddiscussion}.\arabic{table}}
\setcounter{table}{0}

This Section is to introduce our main results concerning the evolution of a massless scalar field on Kerr background. As emphasized earlier distinguished attention will be paid to the angular dependence of the field and to the formation of superradiance. Before presenting our numerical results it is important to justify the reliability of the proposed new method. 

\subsection{Error estimates and convergence}

As emphasized in Section~\ref{numericalevolution} the representation of the basic variables by truncated multipole series can only be `exact' in the case of linear field equations. On the other hand, whenever the evolution equations contain non-linear expressions of the basic variables---with non-trivial angular dependencies---the multipole method becomes inherently perturbative. Nevertheless, it is believed that the error yielded by the truncation of the infinite multipole series remains at a tolerable level provided that the value of $\ell_{\mathrm{max}}$ is kept at a sufficiently high value. In order to demonstrate that this expectation is valid{,} the $\ell_{\mathrm{max}}$ dependence of some estimates on the error and the convergence will be {shown below}. 

Almost all of our simulations were performed by using $\ell_{\mathrm{max}}=12$, nevertheless, in order to be able to determine the convergence rate simulations with $\ell_{\mathrm{max}}=14$, $16$ and $18$ were also performed in the Kerr case with initially co-rotating and counter rotating distributions. Note that as the initial data had pure {quadrupole} character the indicated variation of the value $\ell_{\mathrm{max}}$ had no effect on it.

In what follows the numerical representation---with maximal multipole order $\ell_{\mathrm{max}}$---of a function $f$ will be denoted by $f_{{}_{\ell_{\mathrm{max}}}}$. Assume that $\Delta{\ell_{\mathrm{max}}}$ is some positive integer. 
As a measure of the relative error of the variable $f_{{}_{\ell_{\mathrm{max}}}}$ the quantity
\begin{equation}
\label{relativeerror}
E_{^{}_{\ell_{\mathrm{max}},\Delta{\ell_{\mathrm{max}}}}}(f)=\frac{\Vert f_{{}_{\ell_{\mathrm{max}}}} - f_{{}_{\ell_{\mathrm{max}}+\Delta{\ell_{\mathrm{max}}}}} \Vert}{\Vert f_{{}_{\ell_{\mathrm{max}}+\Delta{\ell_{\mathrm{max}}}}} \Vert}\,
\end{equation}
was applied. Notice that $E_{^{}_{\ell_{\mathrm{max}},\Delta{\ell_{\mathrm{max}}}}}(f)$ monitors the time dependence of the difference of the basic and finer solutions $f_{{}_{\ell_{\mathrm{max}}}}$ and $f_{{}_{\ell_{\mathrm{max}}+\Delta{\ell_{\mathrm{max}}}}}$ relative to the finer one. The norm $\Vert\cdot\Vert$ applied here, and in (\ref{convfact}) below, is the $C^{0}$ norm bounded from above by the second Sobolev norm  $C\Vert\cdot\Vert_{H^{2}_{2}(\mathbb{S}^{2},\mathbb{C})}$, where $C$ is the minimal Sobolev constant associated with the $H^{2}_{2}({\mathbb{S}^{2}},\mathbb{C})\subset C^{0}({\mathbb{S}^{2}},\mathbb{C})$ Sobolev embedding as discussed in \ref{sobolevembedding} and \ref{sobolevconstant}. 
Clearly, $E_{^{}_{\ell_{\mathrm{max}},\Delta{\ell_{\mathrm{max}}}}}(f)\ll 1$ has to hold for reasonable numerical solutions provided that the value of $\ell_{\mathrm{max}}$ is sufficiently large. 

Another useful quantity characterizing the validity of the applied numerical schema is the d'Alembert convergence factor defined for the numerical representation $f_{{}_{\ell_{\mathrm{max}}}}$ as 
\begin{equation}
\label{convfact}
Q_{^{}_{\ell_{\mathrm{max}},\Delta{\ell_{\mathrm{max}}}}}(f)=\frac{\Vert f_{{}_{\ell_{\mathrm{max}}+\Delta{\ell_{\mathrm{max}}}}} - f_{{}_{\ell_{\mathrm{max}}+2\Delta{\ell_{\mathrm{max}}}}} \Vert}{\Vert f_{{}_{\ell_{\mathrm{max}}}} - f_{{}_{\ell_{\mathrm{max}}+\Delta{\ell_{\mathrm{max}}}}} \Vert}\,.
\end{equation}
This quantity measures the local convergence rate centered at $\ell_{\mathrm{\max}}+\Delta{\ell_{\mathrm{\max}}}$. In virtue of d'Alembert's criterion {guaranteeing} a sequence to be summable{,} convergence of the numerical solution in $\ell_{\mathrm{max}}$ occurs provided that the inequality {
$\limsup_{\ell_{\mathrm{max}}\rightarrow\infty}\left(Q_{\ell_{\mathrm{max}},\Delta{\ell_{\mathrm{max}}}}(f)\right)<1$} holds.

A related quantity---which is useful in quantifying the appropriateness of the numerical scheme---is the local exponent of convergence defined by the ratio 
\begin{equation}
\frac{\ln\left(Q_{^{}_{\ell_{\mathrm{max}},\Delta{\ell_{\mathrm{max}}}}}(f)\right)}{\Delta{\ell_{\mathrm{\max}}}}\,. 
\end{equation}
As discussed in \ref{tailsumerrorbounds}, whenever a function $f$ is {
$C^{\infty}$} the convergence has to be faster than any polynomial{ in 
$\ell_{\mathrm{max}}$}, and therefore the local exponent of convergence becomes an informative measure 
of convergence rate---which is expected to be constant whenever the convergence is exponential.

On Figure~\ref{convergence} the time and  $\ell_{\mathrm{max}}$ dependencies of the relative error $E_{^{}_{\ell_{\mathrm{max}},\Delta{\ell_{\mathrm{max}}}}}(\Psi)$ 
relevant for the basic variable $\Psi$ are shown for an initially co- or 
counter rotating massless scalar field on a Kerr background with parameters 
$M=1$ and $a=0.99$. As it is clearly visible the relative error is always 
smaller then $10^{-8}$ and it is decreasing while $\ell_{\mathrm{max}}$ is 
increased. If considerations are restricted to the initial part of 
the evolution, i.e., to the part before the wave packets leave the 
computational domain at $t\sim 75$, the relative error does not exceed the 
level $\sim  10^{-12}$. Following this initial, truly dynamical, period eight 
order smaller amplitude quasi-normal ringing and finally an even smaller amplitude 
power low tail decay occur (see, e.g., \cite{RT}). The amplitude of these processes is comparable to the 
applied accuracy of the present simulations which yields a visible increase 
in the relative error. Notice also that the apparent linear hierarchy of the 
graphs---transparent in the applied logarithmic scale in both of the 
subregions---justifies that, indeed, the convergence rate in $\ell_{\mathrm{max}}$ 
is exponential.

\begin{figure}[!ht]
\begin{center}
\includegraphics{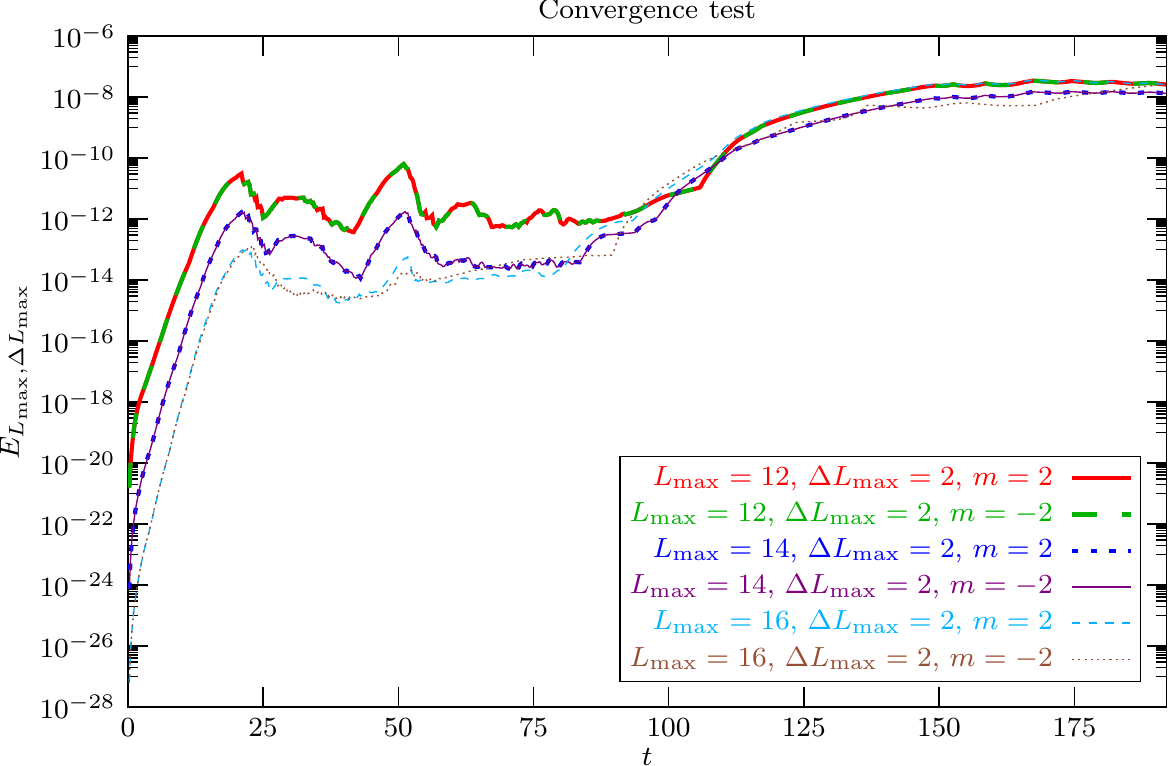}
\caption{\label{convergence}(Color online) 
The coordinate time and ${\ell_{\mathrm{\max}}}$ 
dependence of the relative error $E_{^{}_{\ell_{\mathrm{max}},\Delta{\ell_{\mathrm{max}}}}}(\Psi)$ of $\Psi$ representing the evolution of initially co- or counter rotating massless scalar field with $\ell=2, m=2$ or $\ell=2, m=-2$, respectively, on a Kerr background with parameters $M=1$, $a=0.99$ and for the particular values $\ell_{\mathrm{\max}}=12, 14$ and $16$ are shown. The apparent linear shifting of the error curves on this figure verify that the rate of convergence in $\ell_{\mathrm{max}}$ is exponential.}
\end{center}
\end{figure}

Let us mention here that as the finite differencing part of our 
code, applied in the ``$t-r$'' Lorentzian sector, is exactly the one which, 
along with the AMR part, went through careful and detailed convergence 
tests---and the pertinent results can be found in 
\cite{gridripper,gridripper2}---we would like to recall here only that this 
part of the code is of fourth order accurate as it should be according to the 
implemented numerical scheme.

The computational times as listed in Table~\ref{comptime}---relevant for the same systems as described in connection with Figure~\ref{convergence} and also for the common PC architecture \texttt{AMD Phenom(tm) 2.3GHz CPU}---justify that the required computational resources are affordable, i.e.~the proposed new method is computationally inexpensive.
\begin{table}[!ht]
\begin{center}
\begin{tabular}{c|c|c}
$\ell_{\mathrm{max}}$ & Comp.\ time ($m=-2$) & Comp.\ time ($m=2$) \\
\hline
$12$   & $39582\,\mathrm{sec}$ &  $56374\,\mathrm{sec}$ \\
$14$   & $53274\,\mathrm{sec}$ &  $76089\,\mathrm{sec}$ \\
$16$   & $68900\,\mathrm{sec}$ &  $99248\,\mathrm{sec}$ \\
$18$   & $87514\,\mathrm{sec}$ & $124727\,\mathrm{sec}$ \\
\end{tabular}
\caption{\label{comptime} The computation time of the evolution of initially 
co- or counter rotating massless scalar field with $\ell=2, m=2$ or 
$\ell=2, m=-2$, respectively, on a Kerr background with parameters 
$M=1$, $a=0.99$ within the coordinate time interval $0\leq t \leq 192$ and 
with the particular choices of the values $\ell_{\mathrm{\max}}=12, 14, 16$ 
and $18$. The indicated times were measured by using PCs with architecture 
\texttt{AMD Phenom(tm) 2.3GHz CPU} which justify that the proposed new method 
is computationally inexpensive.}
\end{center}
\end{table}

\subsection{Energy and angular momentum balances}

In addition to the rate of convergence in $\ell_{\mathrm{max}}$ the use of the energy and angular momentum balance relations provides another important consistency check verifying the reliability of the proposed numerical algorithm. 

Recall that the balance laws relate values of energy and angular momentum on portions of $t=const$ hypersurfaces to energy and angular momentum fluxes across the timelike hypersurfaces connecting the edges of them. In particular, the argument goes as follows. Whenever there is a divergence free vector field $J^a$ on a spacetime it can be justified by referring to Stokes' theorem that for a spacetime domain $N$ with boundary $\partial N$ and outward pointing unit normal vector $n_a$ at $\partial N$ the  balance relation 
\begin{equation}
\label{eq.energy}
\int_{\partial N} n_a J^a=\int_{int N} \nabla_a J^a =0\,
\end{equation}
holds. On the other hand, it is well-known that---as the vector fields $\partial_{t}$ and $\partial_{\varphi}$ are Killing vectors on a Kerr spacetime---the contractions $J_E^a=-{T^{a}}_{b}\partial_{t}^{\hskip0.04cm b}$ and {$J_L^a={T^{a}}_{b}\partial_{\varphi}^{\hskip0.04cm b}$}, which are the energy and angular momentum currents, are divergence free, where $T_{ab}$ denotes the energy-momentum tensor of the matter fields. 
In our investigations, $N(t)$ was chosen to possess---in the tortoise Boyer-Lindquist coordinates---the form of the Cartesian product $
[0,t]\times [r_*{}_1,r_*{}_2]\times [0,\pi]\times [0,2\pi]$.

\begin{figure}[!ht]
\begin{center}
\includegraphics{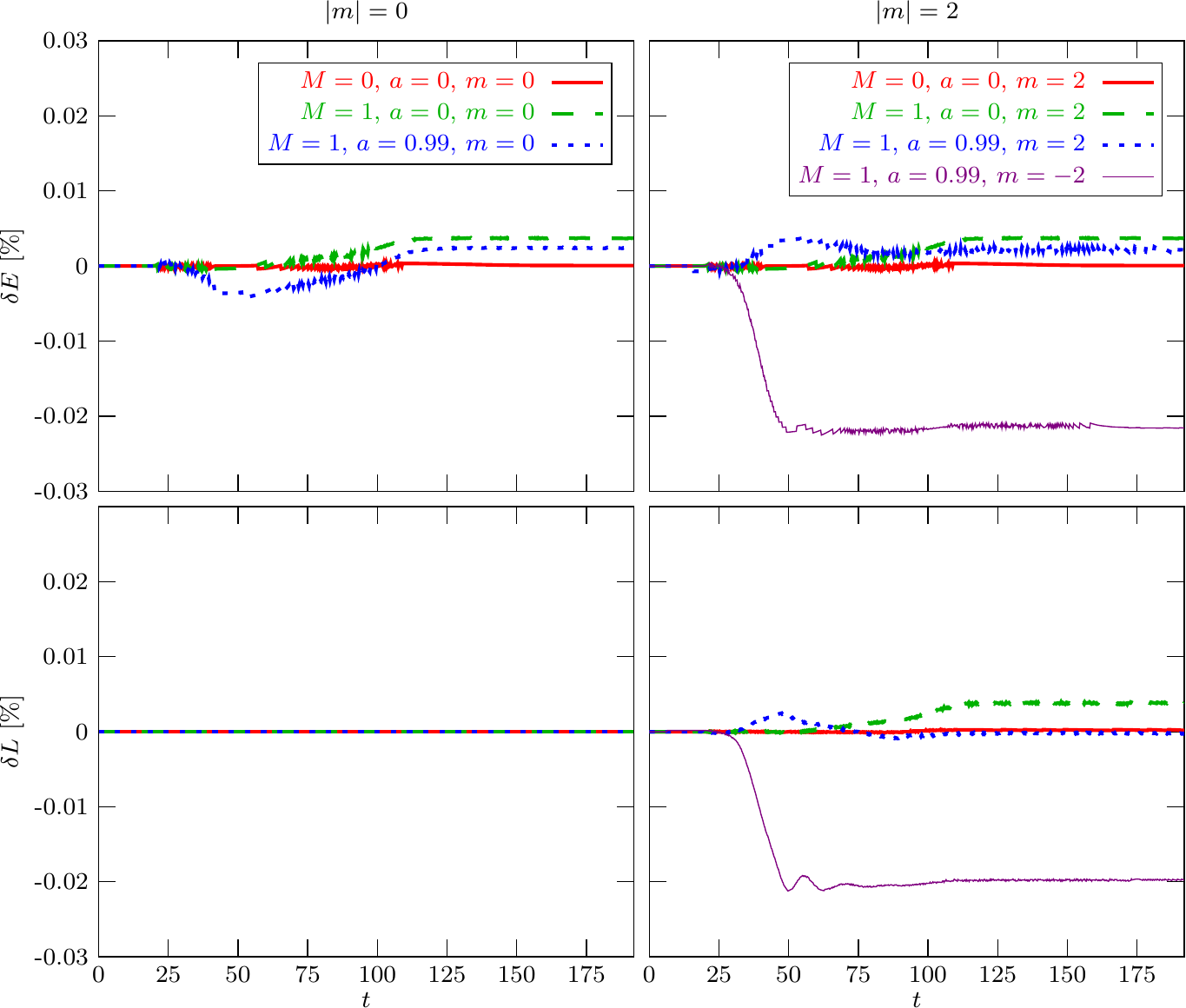}
\caption{\label{conservation}(Color online) 
The time dependence of the relative variation $\delta E=\frac{1}{E_0}\int_{\partial N(t)}\hskip-.0cm n_a J_P^a$ and $\delta L=\frac{1}{L_0}\int_{\partial N(t)}\hskip-.0cm n_a J_L^a$ of energy and angular momentum balances during the evolution of a massless scalar field on Minkowski, Schwarzschild and Kerr background spacetimes with initially co-rotating ($m=2$), non-rotating ($m=0$) and counter-rotating ($m=-2$) pure quadrupole type initial data. The reference values $E_0$ and $L_0$ are the  initial energy and angular momentum contents of the selected parts of the initial data surface $t=0$, respectively.}
\end{center}
\end{figure}

The constant $r_*$ values determining the edges of spatial section of the cylindrical domain of integration $N(t)$---in order to keep some margin from the edges of the computational domain---were chosen to be such that $r_*{}_1=-63$ and $r_*{}_2= 63$ for the Kerr or Schwarzschild cases with $M=1$, whereas $r_*{}_1=0$ and $r_*{}_2= 63$ were used in the Minkowski limit with $M=0$, where $r_*$ reduces to $r$. 

On Figure~\ref{conservation} the time dependence of the relative variation $\delta E$ and $\delta L$ of energy and angular momentum balance relations are shown. Here $\delta E$ and $\delta L$ are defined as
\begin{equation}
\delta E=\frac{1}{E_0}\int_{\partial N}\hskip-.3cm n_a J_E^a \quad {\rm and} \quad \delta L=\frac{1}{L_0}\int_{\partial N}\hskip-.3cm n_a J_L^a\,,
\end{equation}
where $E_0$ and $L_0$ denote the energy and angular momentum of the initial configuration within the spatial region  $r_*\in[r_*{}_1,r_*{}_2]${.}
The graphs on Figure~\ref{conservation} make it transparent that the energy and angular momentum balances hold up to a remarkable precision for the entire evolution. 

Note, finally, that the evaluation of the involved integrals can be done in a straightforward way in context of the spectral method as the integration with respect to the angular degrees of freedom can simply be given as $L^{2}$ scalar products of the basic variables which can be evaluated as outlined in Appendices~
A, B, C and D.
In addition, the integrals with respect to the radial and temporal directions were evaluated by using a fifth order integration scheme to avoid the loss of accuracy of the numerical data yielded by the applied {fourth order} finite difference scheme in the $t-r_{*}$ plane.

\subsection{Angular dependencies of the fields}

After presenting the consistency checks of the applied numerical scheme let us turn to the description of the physical properties of the solutions. In this Section our main concern is the angular dependence of the evolving scalar field. 

To have some hints regarding the dynamics of a massless scalar field with co-rotating quadrupole type data on Kerr spacetime the energy density and the momentum current distributions are shown on Figure~\ref{3d} on the initial data surface, at $t=0$, and on an intermediate time level surface after a {scattering} of the inward falling radiation has happened, at $t=48$. It is visible that, in spite of the fact that the initial data was also fine tuned to be maximally superradiant, the dominant part of the outgoing radiation leaves the central region without indicating the slightest preference of directions close to the axis of rotation. 

On Figure~\ref{outflow} the time dependence of both the total integrated fluxes of the radiated energy and angular momentum through the $r_*=63$ sphere and the fluxes of the radiated energy and angular momentum integrated on the caps of sphere yielded by the intersection of a double right circular rotationally symmetric cone with apex angle $2\theta=\pi/3$ and the $r_*=63$ sphere is shown for various configurations. 

\begin{figure}[!ht]
\begin{center}
\includegraphics{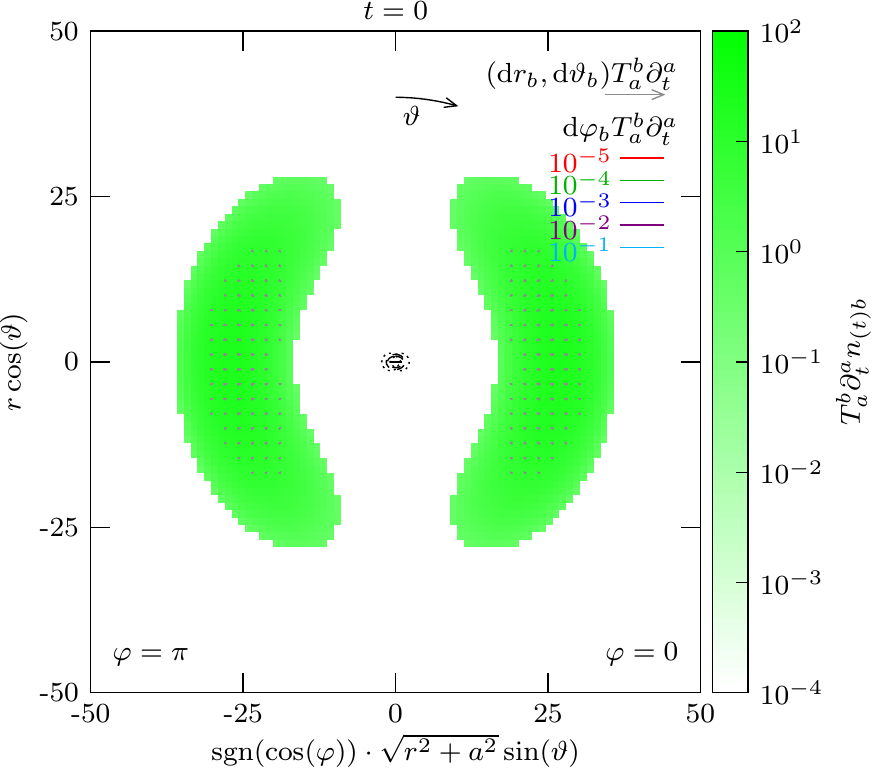}
\vspace*{2mm}
\includegraphics{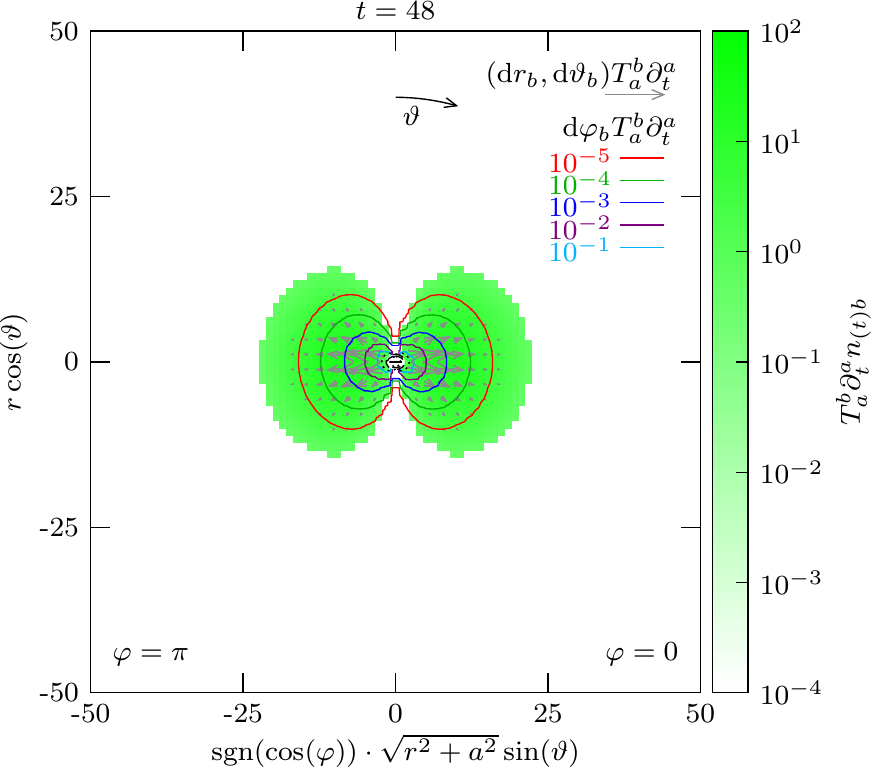}
\caption{\label{3d}(Color online) The spatial distribution of energy density 
$n_{(t)b}{T_{a}}^{b}\partial_{t}^{\hskip0.04cm a}$ and the $r, \vartheta, \varphi$ 
components of the energy current ${T_{a}}^{b}\partial_{t}^{\hskip0.04cm a}$ 
are shown at $t=0$ (top panel) and at $t=48$ (bottom panel) for a massless 
initially co-rotating scalar field on a Kerr background with parameters $M=1$ 
and $a=0.99$. Note that only the sections corresponding to the azimuthal slices 
$\varphi=0$ and $\pi$ are plotted. The initial data was fine tuned to be 
maximally superradiant and by $t=48$ a scattering has 
already happened. The energy density, 
$n_{(t)b}{T_{a}}^{b}\partial_{t}^{\hskip0.04cm a}$, is indicated by the 
color map while the $\D r_{b}{T_{a}}^{b}\partial_{t}^{\hskip0.04cm a}$ and 
$\D\vartheta_{b}{T_{a}}^{b}\partial_{t}^{\hskip0.04cm a}$ components of the 
energy current tangent to the $\varphi=0$ and $\pi$ plane are indicated by 
arrows, whereas the azimuthal component  
$\D\varphi_{b}T_{a}^{b}\partial_{t}^{\hskip0.04cm a}$ is depicted by isocurves. 
The location of the singularity, the event horizon and the ergosphere is also 
indicated on central parts of the plots. Note that for the 
{sake} of simplicity the quantities indicated are 
given by referring to the Boyer-Lindquist coordinates.}
\end{center}
\end{figure}

It can be seen that in the case of rotating initial configuration (with $m\neq0$) the outgoing radiation is suppressed in the vicinity of the axis of rotation. An effective evaluation of the flux integrals on the caps of the sphere yielded by the intersection of a double right circular rotationally symmetric cone in the spectral framework requires additional technicalities which are described in details in \ref{radiationanisotropy}.

\begin{figure}[!ht]
\begin{center}
\includegraphics{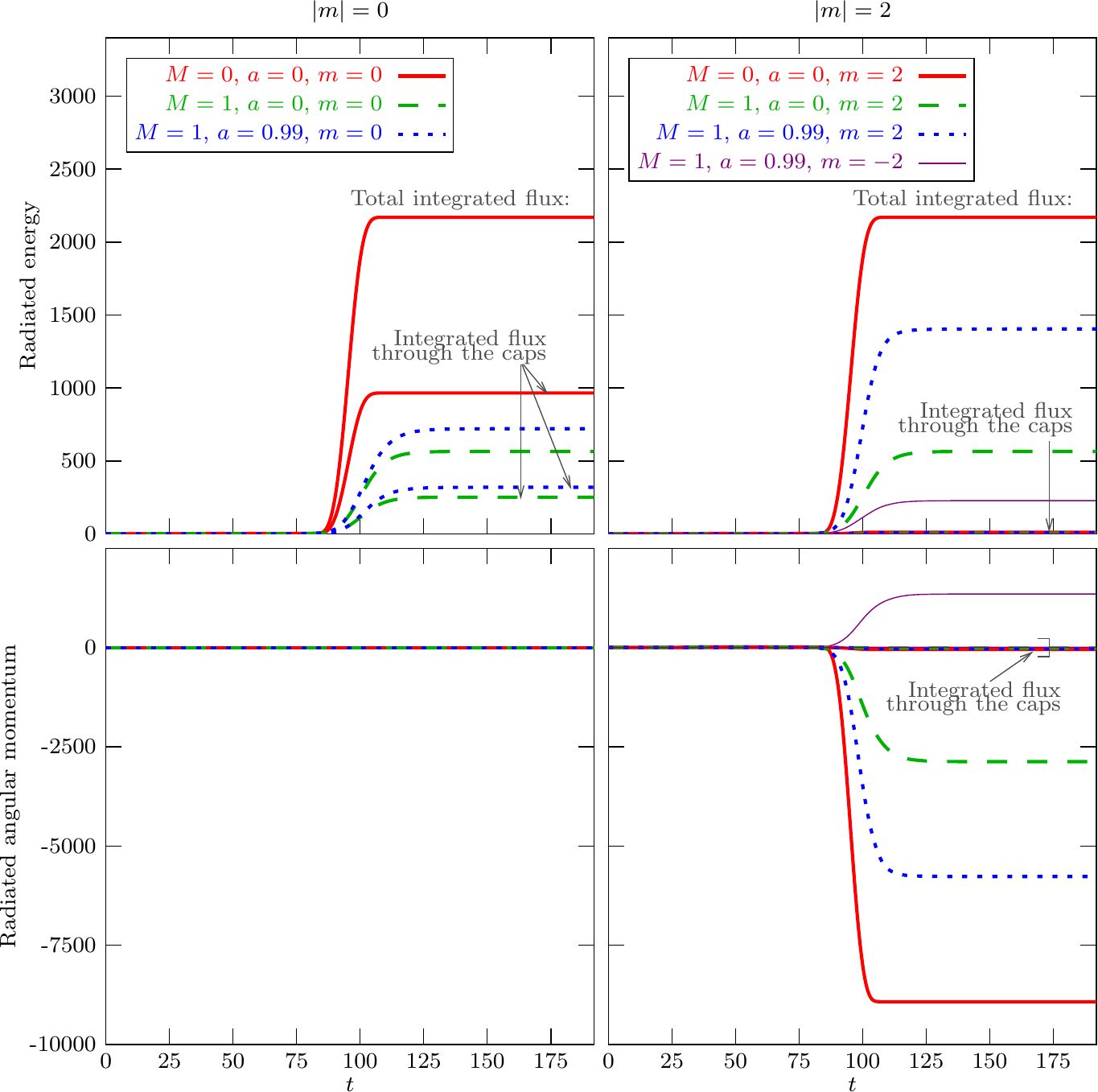}
\caption{\label{outflow} (Color online) 
The time dependence of the integrated energy and angular momentum fluxes 
during the evolution of an initially quadrupole type co-rotating ($m=2$), 
non-rotating ($m=0$) and counter rotating ($m=-2$) massless scalar field on 
Minkowski, Schwarzschild and Kerr spacetimes is shown. Besides the total 
integrated fluxes the fluxes integrated on the disjoint caps of the sphere 
with radius $r_*=63$ yielded by the intersection of a double right circular 
rotationally symmetric cone with apex angle $2\theta=\pi/3$ and the sphere 
are plotted. Note that in the Kerr case the co-rotating initial data with 
$m=2$ was fine tuned to be maximally superradiant.}
\end{center}
\end{figure}

It still remains really cumbersome to extract some insight concerning the anisotropy of the outgoing radiation simply by inspecting plots of the type depicted by Figs.~\ref{3d} and \ref{outflow}. Therefore it is important to have a clear measure of anisotropy. Assume that $J^a$ is a conserved current and consider a ball $\mycal{B}(\bar r_*)$ of radius $r_*=\bar r_*$ in the outer region of the computational domain. Denote by $\mycal{B}(\bar r_*,\theta)$ the disjoint union of the two caps yielded by the intersection of a double right circular rotationally symmetric cone with apex angle $2\theta$ and the ball $\mycal{B}(\bar r_*)$. Clearly then $\mycal{B}(\bar r_*)=\mycal{B}(\bar r_*,\pi/2)$. Finally, denote by $\left.X_{\mathrm{out}}^J\right\vert_{\theta}(t)$ the integral $\int_{[0,t]\times \mycal{B}(\bar r_*,\theta)} n_{r_*}{}_a J^a$, where $n_{r_*}{}_a$ stands for the unit form field normal to the hypersurface $[0,t]\times \mycal{B}(\bar r_*,\theta)$, i.e.~$n_{r_*}{}_a$ points to the increasing $r_*$ direction.

Based on the above introduced quantities, as a measure of anisotropy, we may use then the expression 
\begin{eqnarray}
\label{anisotropy}
{\mycal{A}^J(t,\theta)=\frac{\frac{1}{4\pi(1-\cos \theta)}}{\frac{1}{4\pi}}\frac{\left.X_{\mathrm{out}}^J\right\vert_{\theta}(t)}{\left.X_{\mathrm{out}}^J\right\vert_{\pi/2}(t_{\mathrm{max}})}}
\end{eqnarray} 
which is nothing else but the ration of the time dependent angle average of the integrated flux of the current $J^a$ through the two caps, $\mycal{B}(\bar r_*,\theta)$, of the ball $\mycal{B}(\bar r_*)$ of radius $r_*=\bar r_*$ located at the north and south poles and of the angle average of the total integrated flux through the entire ball $\mycal{B}(\bar r_*)$.  

Notice that, for any fixed $\theta\in(0,\pi/2)$ value, whenever the radiation has no anisotropy at all $\mycal{A}^J(t,\theta)$ tends to $1$ as {$t\rightarrow t_{\mathrm{max}}$}, it tends to a value smaller than $1$ if the radiation shows preferences of the directions close to the equatorial plane, while $\mycal{A}^J(t,\theta)$ tends to a value greater than $1$ if the radiation prefers the axial directions. Clearly by choosing $\theta$ to be small the sharp preference of the axis can be tested, whereas by increasing its value the anisotropy can be tested for a wider range of directions around the axis. In all of our investigation $\theta$ was chosen to be $\pi/6$, i.e.~all directions within a right circular rotationally symmetric cone with apex angle $2\theta=\pi/3$ were included. 

On Figure~\ref{scaledoutflow} the time dependence of the energy and angular momentum radiation anisotropies, $\mycal{A}^{J_E}(t,\pi/6)$ and $\mycal{A}^{J_L}(t,\pi/6)$ is shown for Kerr, Schwarzschild and Minkowski background spacetimes. The evolution starts with co-rotating or counter rotating initially field configurations. As before the co-rotating initial data was fine tuned to be maximally superradiant. By the inspection of Figure~\ref{scaledoutflow} the following simple observation can be made.
\begin{itemize}
\item Regardless whether the initial configuration was co-rotating or counter rotating a strong preference of the directions close to the equatorial plane is justified {by the} asymptotic behavior of $\mycal{A}^{J_E}(t,\pi/6)$ and $\mycal{A}^{J_L}(t,\pi/6)$ both of which tend to a value much smaller then $1$. 

\item It is also clearly visible that the properties of the background spacetimes have no noticeable effect on the {evolution} of either $\mycal{A}^{J_E}(t,\pi/6)$ or $\mycal{A}^{J_L}(t,\pi/6)$. 

\item Finally, the fact that the initial data for the co-rotating configuration was fine tuned to have the solution to be superradiant has no effect at all on the anisotropy of the energy and angular momentum distribution of the outgoing radiation.   
\end{itemize}

\begin{figure}[!ht]
\begin{center}
\includegraphics{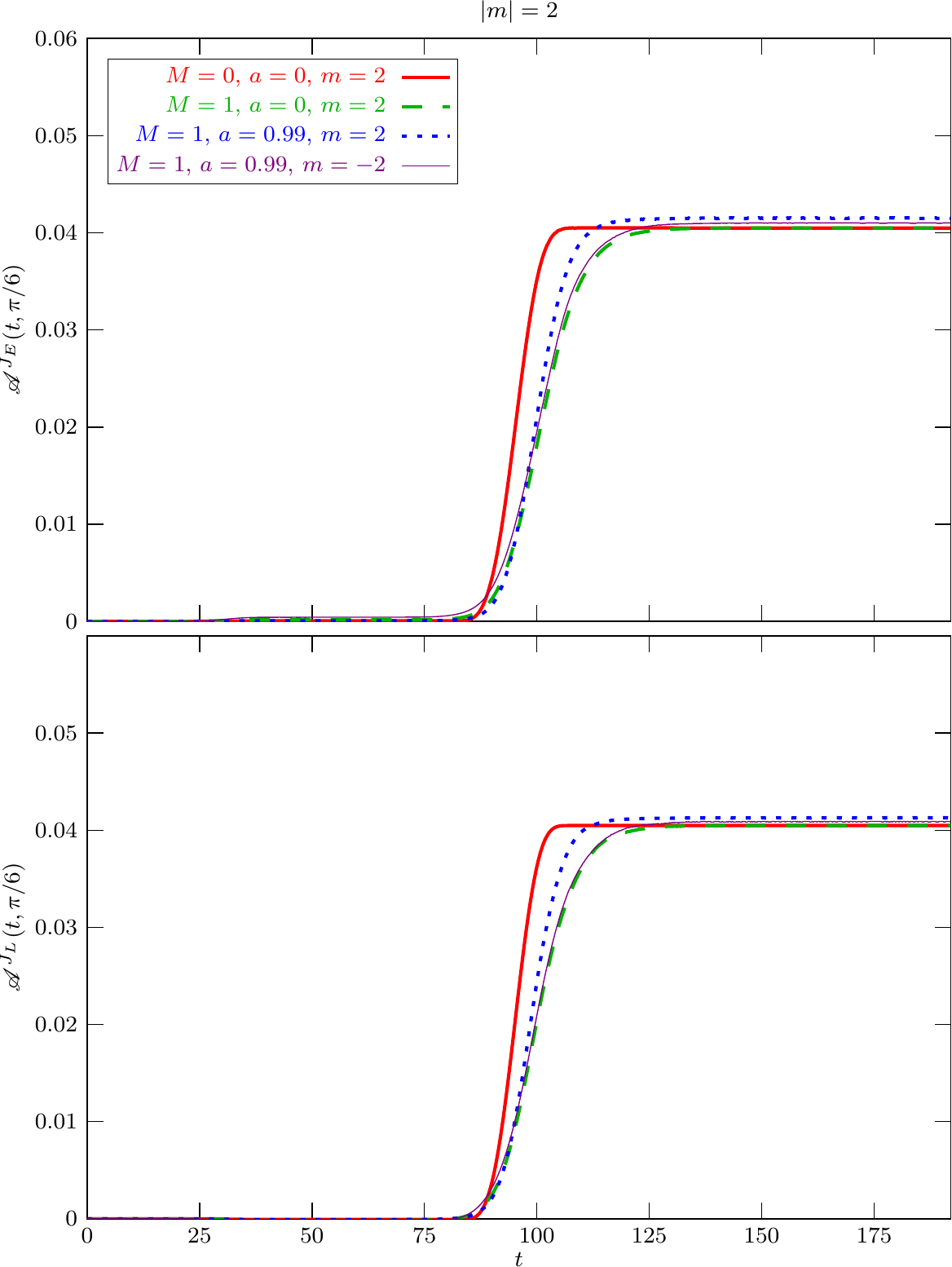}
\caption{\label{scaledoutflow}(Color online) The time dependence of the energy and angular momentum radiation anisotropies, $\mycal{A}^{J_E}(t,\pi/6)$ and $\mycal{A}^{J_L}(t,\pi/6)$, is shown for Kerr, with $M=1$, $a=0.99$, Schwarzschild, with $M=1$, $a=0$, and Minkowski, with $M=0$, $a=0$, backgrounds. The evolution starts with co-rotating or counter rotating initial field configurations. The co-rotating initial data was fine tuned to be maximally superradiant in the Kerr case.}
\end{center}
\end{figure} 

\subsection{Superradiance and nearly perfect reflection}

Based on the observed insensitivity of the anisotropy on the superradiant or non-superradiant character of the initial configuration it turned to be important to understand the reflection and absorption processes during the evolution of the investigated scalar field on black hole backgrounds. 

As it was already indicated in subsection \ref{supercheck} it is important to be sure that the type of initial data chosen there does correspond to be superradiant configurations. To justify that this is indeed the case we determined the temporal frequency spectrum of a numerical solution with initial data parameters $\omega_{0}=\frac{1}{2}m\Omega_{H}$, $r_{*0}=31.823$, $\ell=2$, $m=2$ and with Kerr background parameters $M=1$, $a=0.99$ at the location $r_{*}=14$ which is located towards the black hole with respect to the compact support of the initial data. It is clearly justified by Figure~\ref{spectrum1d} that the spectrum is indeed well contained within the superradiant regime as expected.

\begin{figure}[!ht]
\begin{center}
\includegraphics{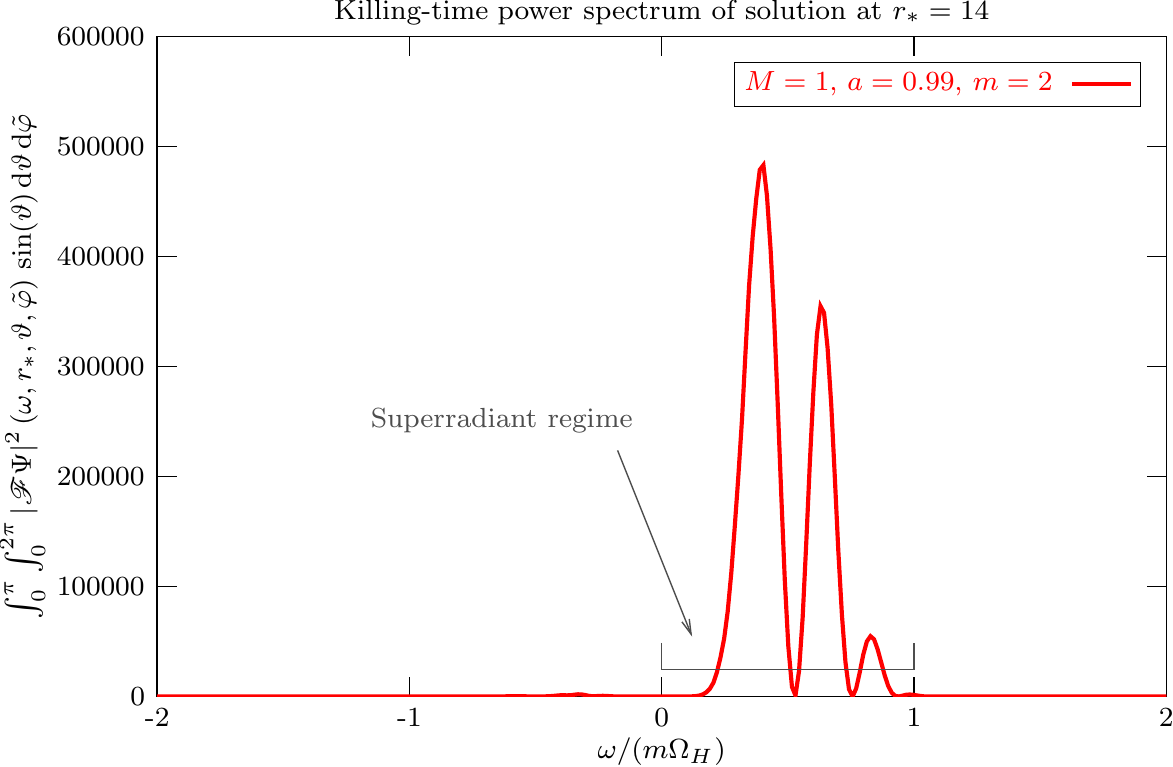}
\caption{\label{spectrum1d}(Color online) The power spectrum in temporal frequency, at the location $r_{*}=14$, 
of an inward traveling wave packet with $\ell=2$, $m=2$ and with initial data having 
leading frequency $\omega_{0}=\frac{1}{2}m\Omega_{H}$. Thereby, the incident wave packet is, as it is expected, to be maximally superradiant. The parameters of the Kerr background were $M=1$, $a=0.99$. Note that 
the spectrum appears to be relatively intact in its superradiant character, i.e.\ 
the solution remains in the desired frequency regime.}
\end{center}
\end{figure}

The generic behavior of the incident wave packets is depicted on Figure~\ref{completereflection}. This figure shows the time dependence of the radial coordinate densities of the energy and angular distributions of the massless scalar field evolving on Minkowski, Schwarzschild and Kerr background spacetimes. The initial data is of quadrupole type and co-rotating or counter rotating 
in the Kerr case, according to the choices $\ell=2$ and $m=\pm2$. 

The radial coordinate density of energy and angular momentum are the quantities $\mycal{E}$ and $\mycal{L}$ with the help of which the energy and angular momentum, $E$ and $L$, on a $t=const$ time level surface can be given as $E=\int_{t=const}\mycal{E}\D r_*$ and $L=\int_{t=const}\mycal{L}\D r_*$, i.e.~in $\mycal{E}$ and $\mycal{L}$ the energy and angular momentum densities are integrated with respect to the angular degrees of freedom, and they also involve the not yet integrated part of the $3$-volume form induced on the $t=const$ time level surfaces. 

We have found that in case of a massive background with $M>0$ for non-superradiant type of initial configurations, as it is expected, considerable part of the incident wave packet gets to be absorbed by the black hole. However, for initial data fine tuned to generate a totally superradiant configuration---contrary to the generic expectations---no energy extraction from the black hole was observed. Instead a nearly total reflection of the wave packet occurs, as it is shown on the bottom right panel of Figure~\ref{completereflection}. Notice the similarities characterizing the evolution of the to be superradiant initial data in the Kerr case (lower right panel) and the evolution of the scalar field on simple Minkowski background (top left panel) with no black hole in the setup. 

\begin{figure}[!ht]
\begin{center}
\includegraphics{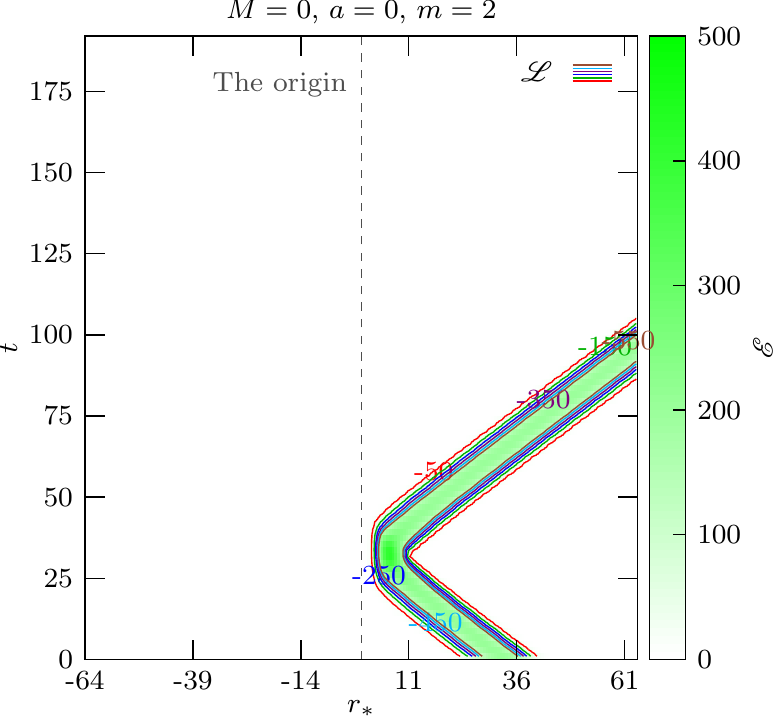}\includegraphics{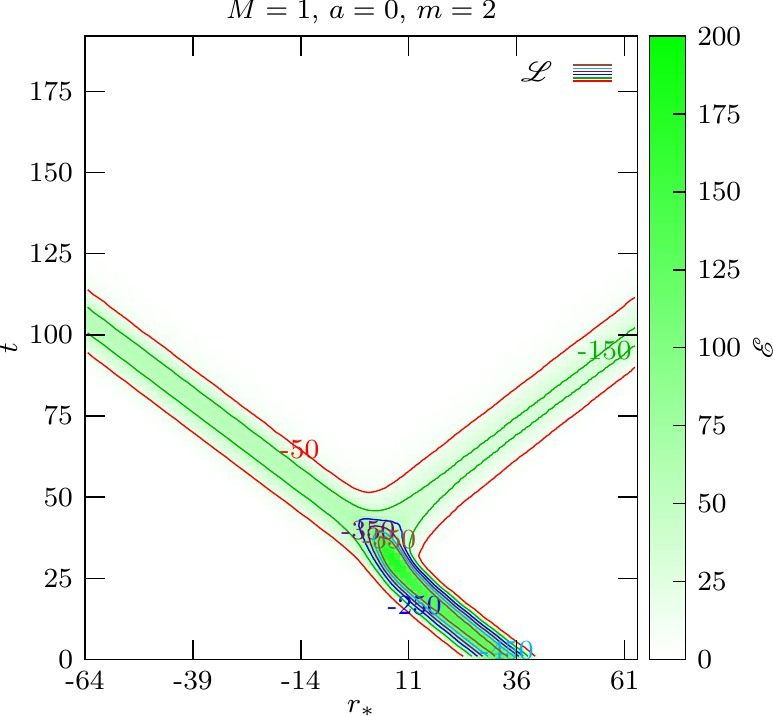}
\includegraphics{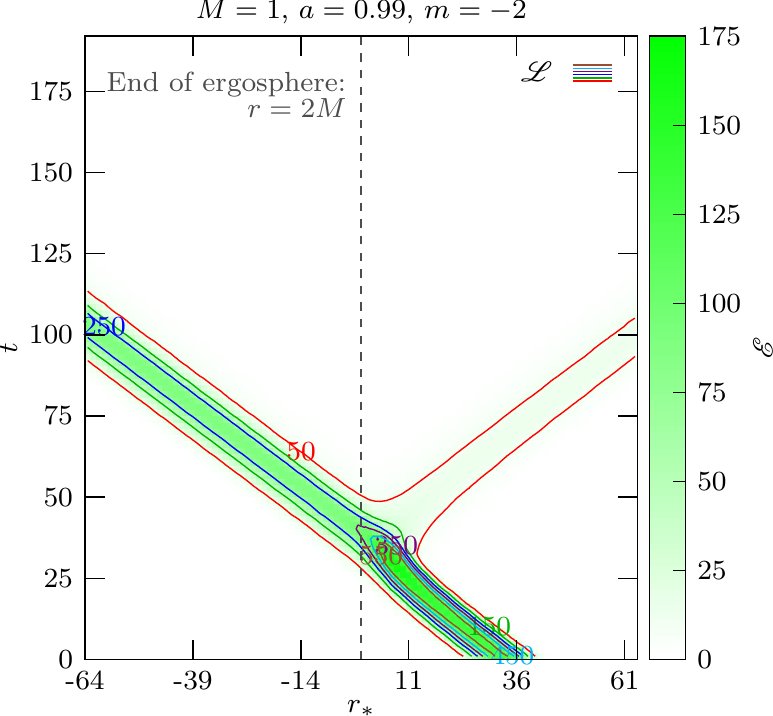}\includegraphics{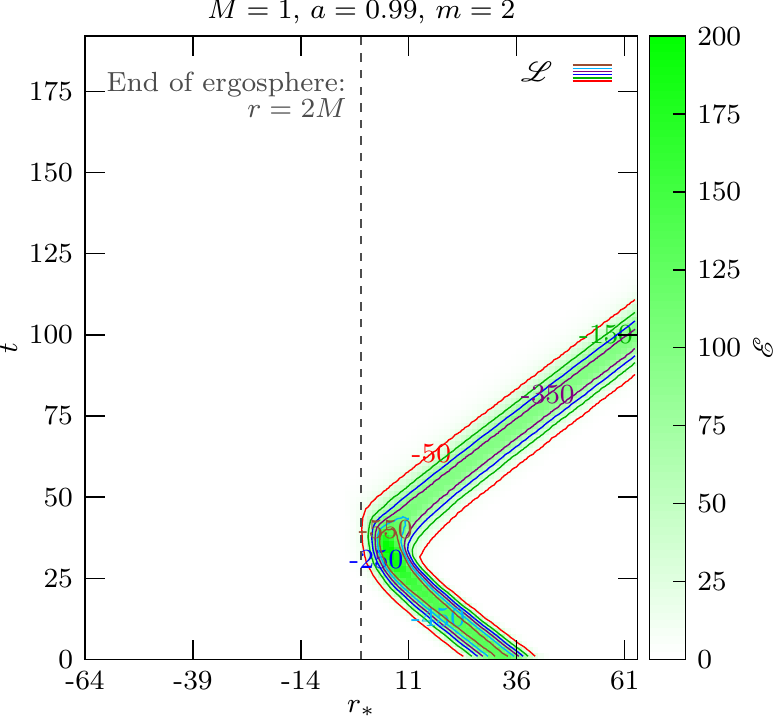}
\caption{\label{completereflection}(Color online) The time dependence of the radial coordinate density of energy and angular momentum, $\mycal{E}$ and $\mycal{L}$ (for their definition see the main text), is shown for Minkowski, Schwarzschild and Kerr background spacetimes. The evolution starts with co-rotating or counter rotating initially field configurations. As before the co-rotating initial data was fine tuned to be maximally superradiant in the Kerr case.}
\end{center}
\end{figure}

In trying to figure out the significance of this nearly perfect reflection recall now that to be able to produce extra energy in the process of superradiance substantial part of the radiation---after submerging into the ergoregion---has to descent towards the event horizon. The lack of energy extraction 
from the black hole can now be understood as the reflection happened before 
the radiation could have reached the ergoregion. 

As our result concerning superradiance is on contrary to the conventional 
expectations it is important to emphasize that the observed phenomenon of 
this nearly perfect reflection, for the part of the wave packet belonging to the 
superradiant regime, was found to be robust with respect to the variation of the 
parameters of the background spacetime and that of the initial data. 
Note also that as our pertinent results appear to be 
inconsistent with the claims \cite{krivan, laguna} it is important to clear up the 
reason beyond these controversial conclusions. In doing so start by 
recalling that 
there is a significant difference between the type of initial data applied in 
\cite{krivan, laguna} and in this paper. While the initial configuration we 
applied is of compact support in \cite{krivan, laguna} the initial data was 
arranged to have non-trivial values everywhere in the ergoregion. 
Moreover, it is claimed in \cite{krivan,laguna} that energy extraction from the black hole does occur. However, in 
our checks---applying horizon penetrating slices as in \cite{RT} and exactly the same type of initial data as in \cite{krivan, laguna}---the to be superradiant character of the field was lost in an extremely short 
period, more importantly, no energy flux leaving the ergoregion could be observed. 
All of these observations should be completed by emphasizing that 
our conclusions are not at all incompatible with claims in 
\cite{starobinskii,yau,DafermosRodianski}. First of all, although in \cite{yau} compactly 
supported data is applied in deriving analytic estimates concerning superradiance the yielded results 
therein are converted, on page 833, to quantitative estimates based on approximations derived 
by Starobinskii \cite{starobinskii}, in spite of the fact that the 
approximations applied in \cite{starobinskii} are not entirely compatible with 
the use of compactly supported initial data. 
What is even more important is that 
the pertinent conclusion in \cite{yau} provides only 
an upper bound for the gained energy which is about $\sim 1\%$ for the case $\ell=m=2$. 
Note also that our numerical findings are 
consistent with the upper bound, $0.05-0.08\%$ for $\ell=2$ modes, on the scale of the energy extraction which we deduced by making use of a detailed single mode analysis (the results of which will be published elsewhere).

In providing some more convincing evidences let us emphasize 
first that the conventional arguments of Misner and Zel'dovich supporting 
the existence of energy extraction are based on the use of individual modes. 
Note, however, that the study of the linear stability problem for Kerr 
spacetimes \cite{DafermosRodianski}, with the application of finite 
energy wave packets, taught us the lesson that statements at the level of individual 
modes need not imply statements for the superposition of infinitely many modes.

Figure~\ref{superreflection} is to provide an additional justification of our main 
result. On this figure the time dependence of the energy extraction coefficient, 
$\frac{{\mycal E}_{0}-{\mycal E}_{\mathrm{out}}}{{\mycal E}_{0}}$ is shown, where 
${\mycal E}_{0}$ and ${\mycal E}_{\mathrm{out}}$, respectively, stand for the initial 
value of the energy and for the integrated energy flux through the ball of radius $r_*=63$ located 
at the outer boundary of the computational domain. As it is visible the graph of
$\frac{{\mycal E}_{0}-{\mycal E}_{\mathrm{out}}}{{\mycal E}_{0}}$ starts at the 
value one and tend to zero from above. Note that the slowly decreasing part, with 
$150 \leq t \leq 1000$, represents only the beginning of a long lasting quasi-normal 
ringing of the scalar field on the black hole background. The energy stored in these ringing 
modes will eternally be also radiated to infinity (see, e.g., \cite{RT}). Since 
$\frac{{\mycal E}_{0}-{\mycal E}_{\mathrm{out}}}{{\mycal E}_{0}}$ does not change 
sign the energy radiated to infinity remain always smaller than ${\mycal E}_{0}$, 
which justify our conclusion that no energy extraction had happened.

Clearly, one could claim that any numerical method has its own limitation which is 
true also in the present case. Nevertheless, in virtue of Figure~\ref{conservation} 
the accuracy of our numerical scheme allows us to put sharp upper bound on energy 
extraction which is $\sim 10^{-3}-10^{-4}$ times ${\mycal E}_{0}$ that is significantly smaller than 
the $\sim1\%$ of ${\mycal E}_{0}$ derived by analyzing individual modes \cite{teukolsky,pressteukolsky,starobinskii}.

\begin{figure}[!ht]
\begin{center}
\includegraphics{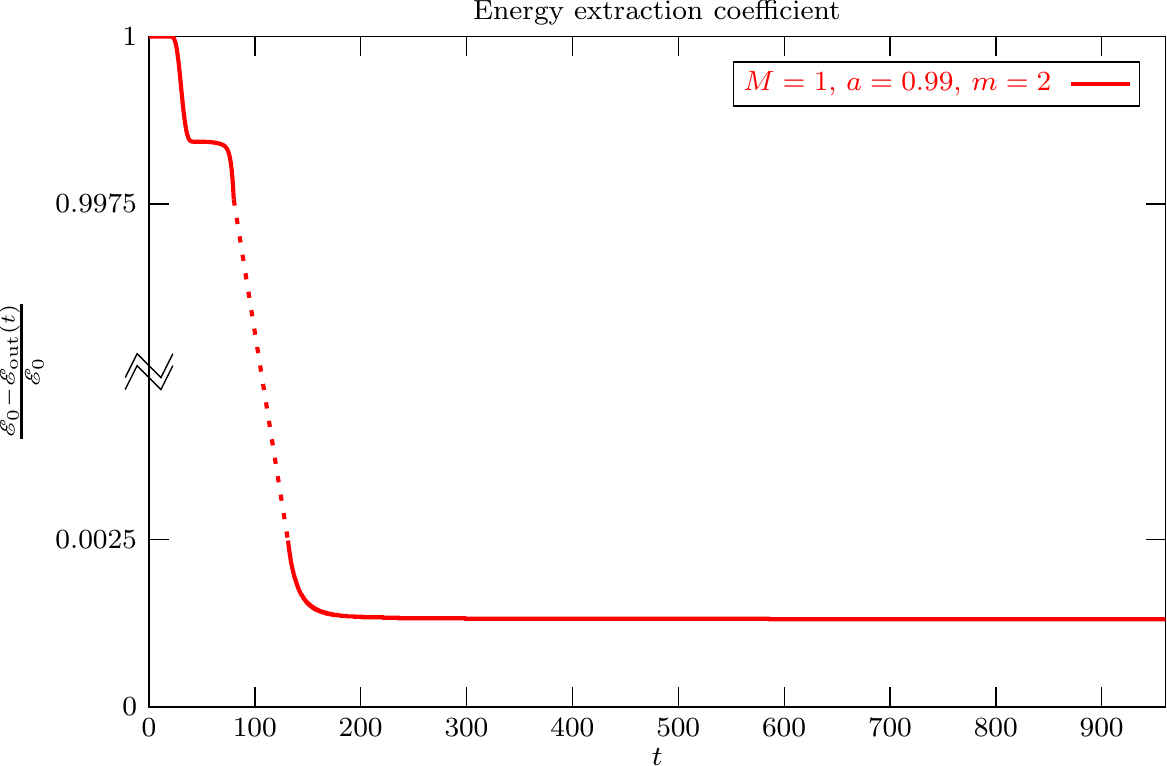}
\caption{\label{superreflection} (Color online) The time dependence of the energy 
extraction coefficient, $\frac{{\mycal E}_{0}-{\mycal E}_{\mathrm{out}}}{{\mycal E}_{0}}$ 
is shown for quadrupole type initial data, where ${\mycal E}_{0}$ and ${\mycal E}_{\mathrm{out}}$, 
respectively, stand for the initial value and for the integrated energy flux through the 
sphere of radius $r_*=63$ located outward with respect to the support of the initial data.}
\end{center}
\end{figure}

Let us finally mention that the power spectrum in temporal frequency of the 
solution provides some new insight what happens whenever the to be superradiant wave 
packet approaches the ergosurface. On Figure~\ref{spectrum2d} the $r_{*}$ dependence 
of the power spectrum in temporal frequency of the solution, which had 
been averaged for the angular degrees of freedom, is shown. The initial data is exactly 
the same quadrupole type with $\ell=m=2$ as used to generate the solution depicted on
Figure~\ref{completereflection}. (Note that to determine the proper Fourier transform of the 
solution for all the indicated values of $r_{*}$ it was necessary to evolve the initial 
data both forward and backward in time.) It is clearly visible that the solution stays in the 
superradiant regime not only in the distant region but up to the ergoregion. In addition,
it is also important that the frequency of the involved modes grows up to reaching 
the value $\omega=m\,\Omega_H$ where, in virtue of the relation 
$(\omega- m \Omega_{H})\,\left\vert\mathcal{T}\right\vert^{2} =(1-\left\vert\mathcal{R}\right\vert^{2})\,\omega$, 
inevitably a total reflection has to occur.

\begin{figure}[!ht]
\begin{center}
\includegraphics{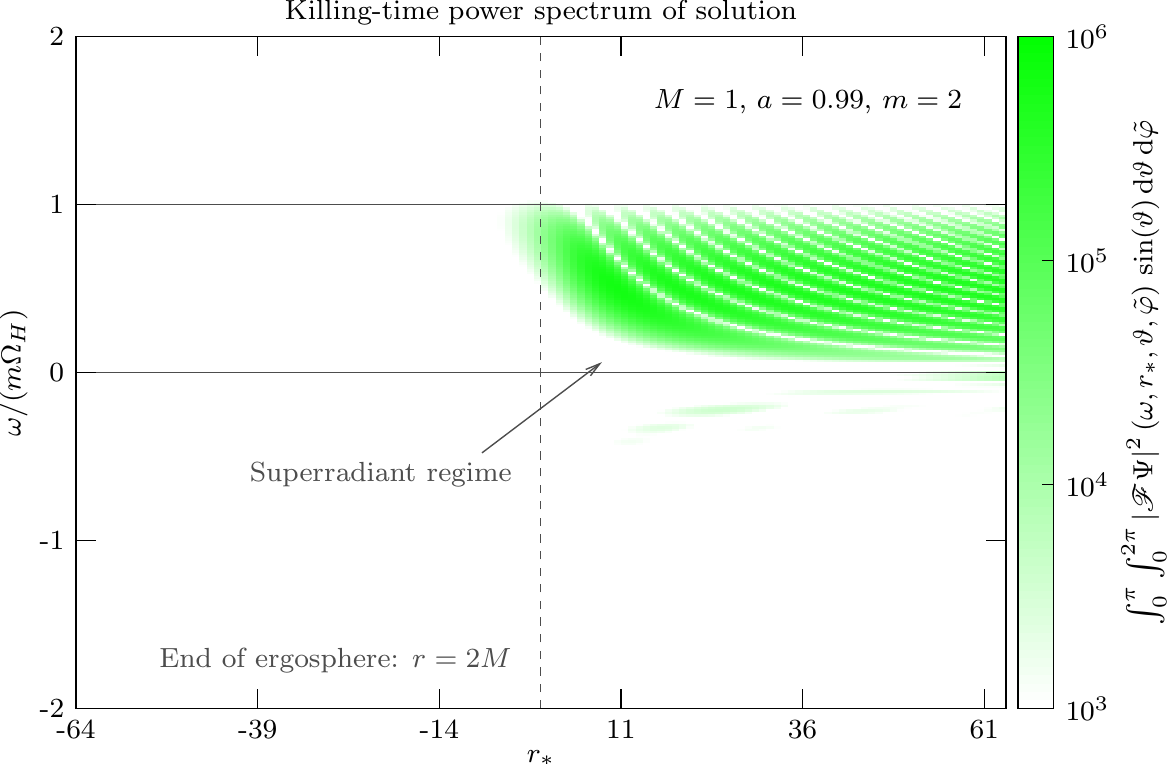}
\caption{\label{spectrum2d} (Color online) The 
$r_{*}$ dependence of the power 
spectrum in temporal frequency of the solution, which had been averaged for the angular 
degrees of freedom, is shown. The applied initial data is exactly the same quadrupole 
type with $\ell=m=2$ as used to generate the solution depicted on Figure~\ref{completereflection}.
The solution stay in the superradiant regime for the entire evolution while the frequency of the 
involved modes grows up to reaching the value $\omega=m\,\Omega_H$ where total reflection has to occur.}
\end{center}
\end{figure}

\medskip

Let us close this Section by commenting the non-negligible reflection visible on the top right panel of Figure~\ref{completereflection} depicting the evolution of the massless scalar field on a Schwarzschild background. One might be surprised by this reflection as intuitively it is tempting to assume that the Schwarzschild black hole would be ready to absorb almost the entire of the incident wave packet. Recall, however, that the angular momentum balance relation does not support the occurrence of such an overwhelming absorption. In addition it is also worth to have a look at the equation on a Schwarzschild background governing the evolution of a pure $Y{}_{\ell}^{m}$ mode. In fact, the $t-r_*$ part of the wave equation for the coefficient $\Psi^m_\ell=\Psi^m_\ell(t,r_{*})$ reads as
\begin{equation}
\label{schwarzschildwave}
\left(\partial^2_{t}-\partial^2_{r_{*}}\right)\Psi_{\ell}^{m} + V_{\ell}(r)\,\Psi_{\ell}^{m}= 0
\end{equation}
with the potential  
$V_{\ell}(r)=\left(1-\frac{2M}{r}\right)\left(\frac{\ell\,(\ell+1)}{r^2}+\frac{2M}{r^3}\right)$. The repulsion---responsible for the reflection of the inward falling radiation and, in turn, leading to the celebrated power decay law of Price \cite{price1,price2}---is transparent on Figure~\ref{potential} depicting the potential $V_{\ell}(r)$ for the $\ell=0,1,2$ cases. 
What is really important here is that the maximum value of the potential $V^{\mathrm{max}}_{\ell=2}\approx 0.24$ is significantly larger than $\omega^2_0\approx0.096$ which, in virtue of the types of arguments contained e.g.\ by Section II of Chapter III in \cite{messiah}, can be used to justify the observed scale of reflection. 

\begin{figure}[!ht]
\begin{center}
\includegraphics{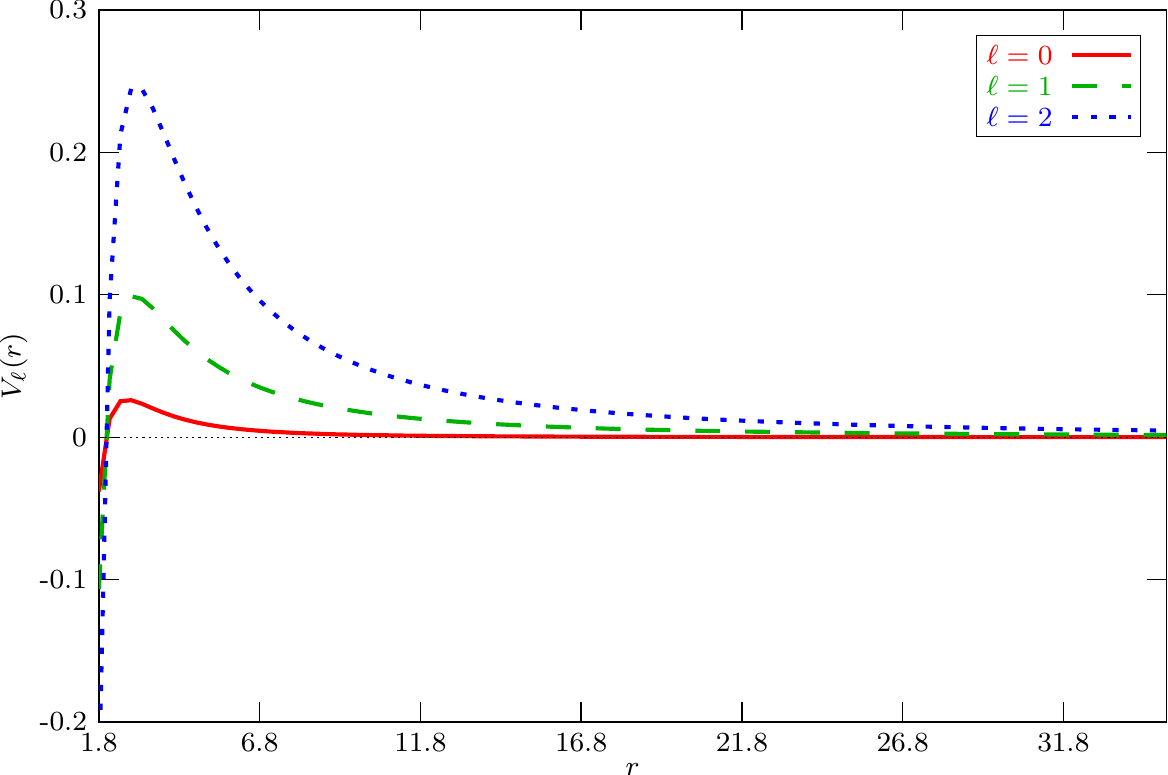}
\caption{\label{potential} (Color online) 
The potential $V_{\ell}(r)$ involved in 
the wave equation over Schwarzschild spacetime 
is plotted for values $\ell=0,1,2$. 
The potentials have their common zero value at the horizon, 
$r=2M$, while they attain their maximum close to but on the domain of outer 
communication side of the event horizon. 
The occurrence of the partial reflection on a Schwarzschild black hole 
background can be understood by taking 
into account the repulsing character of this potential, along 
with the fact that the maximum value of the potential $V^{\mathrm{max}}_{\ell=2}\approx 0.24$ 
is significantly larger than $\omega^2_0\approx0.096$.
}
\end{center}
\end{figure}

\section{Summary}
\label{summary}
\renewcommand{\theequation}{\ref{summary}.\arabic{equation}}
\setcounter{equation}{0}
\renewcommand{\thefigure}{\ref{summary}.\arabic{figure}}
\setcounter{figure}{0}
\renewcommand{\thetable}{\ref{summary}.\arabic{table}}
\setcounter{table}{0}

Our main concern in writing up this paper were at least two folded. On the one hand, we intended to introduce the generic setup of a method that is expected  to provide a powerful new tool in studying the problem of time evolution of non-linear dynamical systems in four-dimensional spacetimes. On the other hand, we applied the introduced new method to study the evolution of a specific dynamical system. More precisely, the evolution of a massless scalar field on a fixed Kerr spacetime was investigated such that distinguished attention was paid to the angular distribution of the evolving field and to the occurrence of superradiance. 

In spite of the fact that the mathematical background of the introduced new method---which is mainly contained by the appendices---makes it to be applicable to dynamical systems the time level surfaces of which can be foliated by a one-parameter family of codimension two surfaces which are conformal to a compact Riemannian manifold $\mycal{C}$ without boundary, in most of the cases with time level surfaces possessing the topology of $\mathbb{R}^3$ or  $\mathbb{R}\times\mathbb{S}^2$ it suffices to assume that $\mycal{C}$ is topological two-sphere $\mathbb{S}^2$. 

One of the main advantages of the new method is that it is fully spectral---not pseudospectral---in the angular directions. Thereby, whenever the basic variables are guaranteed to be at least of class $C^2$ in the angular directions, the spectral components---even though non-linear expressions of the basic variables are involved---can be evolved without their steady pointwise evaluations. Accordingly, the angular degrees of freedom---directions tangential to $\mycal{C}$---are treated by applying $L^2$ expansions of the basic variables in terms of the eigenfunctions of the Laplace operator on $\mycal{C}$. The corresponding expansion coefficients of a basic variable are evolved in the transverse 1+1 dimensional spacetime directions by making use of the method of lines based on a fourth order finite difference numerical scheme such that the adaptive mesh refinement (AMR) is also incorporated. 
 
The main advantages associated with the use of the proposed new method are:
\begin{itemize}
\item the coordinate singularities associated with the angular differential operators are treated in a fully analytic way

\item all the non-linear operations such as multiplication of the basic variables or the division of an expression by a nowhere vanishing variable---the latter can traced back to multiplication with the help of the Neumann series expansion---can be treated within the spectral representation without steady pointwise evaluation of these expressions

\item a very effective treatment of an origin of the time level surfaces can also be done by making use of fully analytic considerations. 
\end{itemize}
On some of the limitations of the proposed numerical method:
\begin{itemize}
\item The angular resolution, determined by the prefixed global value of 
$\ell_{\mathrm{max}}$, is not adaptive in the present implementation of our PDE solver.
This may lead to the loss of accuracy in case of dynamical systems developing 
highly variable angular dependencies. This drawback, however, could be cured, if necessary, 
by endowing $\ell_{\mathrm{max}}$ with a modest `$t-r$' dependence during the evolution, 
which could be controlled by a predefined tolerance in the measure of the tail sum error 
(see \ref{tailsumerrorbounds} for its detailed description).
\item It should also be mentioned here that in case of dynamical systems 
with turbulent instabilities, i.e.\ for configurations which, at certain locations, require 
a simultaneous and substantial increase of both radial and angular resolutions, our method 
(such as many others) may turn to be highly suboptimal as even an adequate local splitting 
$\mathbb{R}^{2}\times \mycal{C}$ may not be found to such a system.
\end{itemize}

A systematic self contained presentation of the mathematical background of the applied results and the implemented elements of the spectral method can be found in the appendices.  

\medskip

The introduced new numerical method is used to study the time evolution of a massless Klein-Gordon field on a fixed Kerr black hole spacetime. We would like to emphasize that this dynamical system is already complex enough to explore the main technical elements of the proposed numerical framework. In particular, two of the most important ones, i.e.~the evaluation of various multilinear expressions and the division by a function based on the use of spectral method were applied in carrying out all of our simulations. 

As all the multipole expansion series are truncated at certain finite order{, 
$\ell_{\mathrm{max}}$,} the proposed numerical method is perturbative. By studying the error introduced by the involved approximations it was verified that the error can be kept at a tolerable low level by applying sufficiently many members of the multipole expansions. In addition, a suitable notion of convergence was also introduced, based on d'Alembert's criterion guaranteeing the summability of sequences. It was justified that the convergence is exponential in the value of $\ell_{\mathrm{max}}$. 

\medskip

In studying the time evolution of a massless Klein-Gordon field on a fixed Kerr black hole distinguished attention was paid to the precise characterization of the angular dependence of the outgoing radiation, as well as, to the development of superradiance. Our main related results are as follows.
\begin{itemize}
\item There are attempts (see, e.g.~\cite{williams, gariel}) aiming to provide a simple and viable physical explanations of high energy collimated matter streams originating from compact astrophysical objects by applying Penrose process, or  superradiance. On contrary to the underlying speculations we have found that the outgoing radiation has no preference at all of the axial directions regardless whether the initial configuration was to be superradiant or a more generic type. 

\item In studying superradiance---in the particular case of a massless scalar 
field---we intended to investigate its formation by using an incident wave 
packet, which was fine tuned to have the 
largest possible part of its frequency content within the 
superradiant regime, from the outer region of the domain of 
outer communication. On contrary to the general expectations, we found that 
instead of the expected scale of energy extraction from the black hole the 
incident wave packet failed to reach the ergoregion. Thereby, instead of 
superradiance, as an interesting phenomenon a nearly total 
reflection occurs before the ``to be'' superradiant part of the incident wave 
packet reaches the ergoregion. A complete characterization of this 
phenomenon definitely exceeds the scopes of the present paper. We plan to 
carry out its further investigations, involving possibly more generic type 
of initial data and, more importantly, a mixture of analytic and numerical 
techniques.
\end{itemize}

Concerning the required computational resources it is also important to be emphasized that various implementations of pseudospectral methods even in case of the study of time evolution of a massless Klein-Gordon field on a fixed Kerr black hole require powerful computers and parallel computing. As opposed to these the proposed new method was found to be very effective in the sense that time evolutions of the very same dynamical systems could be done within reasonable computational time on a stand-alone average personal computer without making use of parallel computing. We would also like to mention that the GPU implementation of the proposed method, which is under development, promises a further significant boost in the optimal use of the computational resources which will be inevitable in studying the evolution of more complicated non-linear dynamical systems. 

Finally, we would like to emphasize that our code GridRipper{, together} with the implementation of the of time evolution of a massless Klein-Gordon field on a fixed Kerr background {and other examples,} is {made} available for public use and it can be downloaded from \cite{gridripper3}.

\ack
We thank G.~Z.~T\'oth for useful discussions and the organizers of the workshop(s) on ``Quantitative Studies of Nonlinear Wave Phenomena'' held at Ervin Schr\"odinger Institute---where parts of the reported results were developed---for their invitation and hospitality.
The authors should also like to thank to projects maintaining the HunGrid VO of the LHC Computing Grid and a standalone cluster, both located at the computer center of {Wigner RCP}, for the access of their computational resources.
This research was supported in part by OTKA grant K67942.

\appendix

\section{Theory of Multipole Expansion of Multilinear Expressions}
\label{sobolevembedding}

Let us start by recalling some of the basic notions and 
notations{ that} 
we shall apply throughout the  succeeding appendices. By a Riemannian
manifold  $M$ we shall always mean an $n$-dimensional paracompact 
manifold 
{(possibly with boundary)} of differentiability class $C^r$ endowed with
Riemannian metric $g${,} which will  usually be suppressed. We shall denote by
$L^{2}(M,\C)$ the   Hilbert space of square-integrable complex valued
functions\footnote{Here in a precise formulation instead of `functions'
function equivalence  classes should be used, where two Lebesgue measurable
functions are considered to be equivalent if they  are almost everywhere
equal. Referring to these function equivalence  classes as simply
`functions' is a common practice in functional analysis which  we also
shall use.}  on a Riemannian manifold $M$, and by $C^{l}(M,\C)$
the vector space of  $l$-times continuously differentiable complex
valued functions  on $M$, where $0\leq l\leq r$. Similarly, $C^{l}_{b}(M,\C)$
stands for the Banach space of bounded $C^{l}(M,\C)$ functions on $M$ equipped
with the $C^{l}$ supremum norm, while $C^{l}_{\infty}(M,\C)$  denotes the
Banach subspace of $C^{l}_{b}(M,\C)$ functions on $M$ which have zero limit
at infinity.\footnote{A 
$C^{l}_{b}(M,\C)$ function is said to have zero 
limit at the infinity, if for any monotonously growing sequence of compact 
sets covering $M$, the $C^{l}$ supremum norm of $f$ over the complement 
of the compact sets tends to zero.}

\smallskip

Consider now the vector space of $r$-times weakly differentiable complex
valued functions over $M$. The Sobolev norm of a function $f$ belonging to
this linear space is defined as
\begin{equation}
\label{sobolevnorm}
\Vert f\Vert_{H_{r}^{2}(M,\C)}=\sqrt{\sum_{l=0}^{r}
  \int_{M}\vert\nabla^{(l)}f\vert^{2}}\,,
\end{equation}
were $\vert\cdot\vert$ denotes the pointwise norm generated by the Riemann
metric $g$, $\nabla$ denotes the Levi-Civita covariant derivation determined
by $g$, while the associated volume form is suppressed in the applied
notation. It is important to keep in mind that the highest, 
{%
$r^{\mathrm{th}}$}-order,
derivation in the above formula is required to be defined in the weak sense.
The subset of the vector space of $r$-times weakly differentiable complex
valued functions over $M$ comprised by elements whose Sobolev norm
$\Vert\cdot\Vert_{H_{r}^{2}(M,\C)}$ is finite is called 
the Sobolev space
and it is  denote{d} by $H_{r}^{2}(M,\C)$. The Sobolev space together with its
norm forms a Hilbert space{,} since its norm is generated by an inner  product,
and also it is complete with respect to this norm.

Let $n$, $r$ and $k$ be non-negative integers, such that $r>\frac{n}{2}+k$.
The classical result of Sobolev embedding  theorem (see e.g.\ \cite{sobolev})
asserts then that the relation
\begin{equation}
\label{sobolevinclusion}
H_{r}^{2}(\R^{n},\C)\subset C^{k}_{\infty}(\R^{n},\C)
\end{equation}
holds, and also that there exists a positive real constant $C$ such that the  
inequality
\begin{equation}
\label{sobolevinequality}
\Vert\cdot\Vert_{C^{k}_{\infty}(\R^{n},\C)}\leq C\,\Vert\cdot\Vert_{H_{r}^{2}(\R^{n},\C)}
\end{equation}
is satisfied. These assertions are also known to hold \cite{sobolev} if
$\R^{n}$ is  replaced by a compact subset in $\R^{n}$. 
In the next  part of this appendix, our aim is to provide a simple and  self-contained
justification of the fact that the  Sobolev embedding theorem may also be
applied in case of compact Riemann manifolds.\footnote{Note, however, that there are non-compact Riemann
manifolds such that the Sobolev embedding theorem does not apply to them. }
For yet another alternative reasoning see, e.g., Ref.\,\cite{laplacecomp}.

\smallskip

In doing so consider first a finite dimensional real vector bundle, $W(M)$, of
differentiability class $C^r$ over a manifold  $M$ that belongs to
the same differentiability class. Then, a pointwise mapping of
$C^{r}$-sections of $W(M)$ onto  $C^{0}$-sections of $M\times\R$ will be
called to be a $C^{r}$-\emph{norm field} if its pointwise restrictions to the
fibers give rise to norms.\footnote{A more adequate way of formulating this definition is that 
a $W(M)\rightarrow M\times\R$ $\;C^{0}$ fiber bundle homomorphism is a $C^{r}$ 
\emph{norm field} if its restrictions to the fibers are norms.}
Indeed, every finite dimensional
$C^{r}$-vector bundle over a paracompact manifold $M$ admits $C^{r}$-norm
fields. To see this recall that a norm field can always be defined locally
over a coordinate chart,  e.g.{,}  by taking in every point the natural
Euclidean norm defined by a trivialization. These locally defined norm fields
may{,} then{,} be sewn together by making use of a partition of unity subordinate
to a locally finite collection of coordinate charts on $M$. The norm fields
over a  vector bundle are equivalent as it is justified by the following
lemma.

\begin{Lem}
\label{equivalence}
If $\vert\cdot\vert$ and $\vert\cdot\vert^{'}$ are $C^{r}$  norm fields over
$W(M)$, then there exists a positive real $C^{r}$ field $C$, such that
$\vert\cdot\vert^{'}\leq C\vert\cdot\vert$.
\end{Lem}
\begin{proof}
The proof is based on the paracompactness of $M$ and on the equivalence of
norms on a finite dimensional vector space.

To start of consider a locally finite atlas $\mycal{A}=\{(U_{i},\varphi_{i}) \,\vert\, i \in I\}$ of $M$ with partition of unity 
{%
$\{\mathcal{F}_{i} \vert i\in I\}$}
such that each
$U_{i}$ has compact closure, denoted  by $\overline{U_{i}}$, in $M$. Assume
that the dimension of the fibers of $W(M)$ is $N$.  Let us fix a
trivialization $\left(e_{i,j} \,\vert\, j \in \{1,\dots,N\} \right)$ of $W(M)$
over each  particular chart $(U_{i},\varphi_{i})\in\mycal{A}$.

As a consequence of the equivalence of norms on a finite dimensional vector
space, for any $p \in M$ there exists a positive number $c_{p}$, such that
$\vert\cdot\vert^{'}_{p} \leq c_{p} \vert\cdot\vert_{p}$.  Furthermore,
$c_{p}$ may be chosen to be  $\sup\limits_{s_{p} \in
  W_{p}(M)\setminus\{0_{p}\}} \frac{\vert s_{p} \vert^{'}_{p}}{\vert s_{p}
  \vert_{p}}$.

By making use of the trivialization $\left(e_{i,j}\right)$ of $W(M)$ over
$(U_{i},\varphi_{i})$ it can be verified immediately that 
\begin{eqnarray}
\label{equivalenceequation}
\sup_{p \in \overline{U_{i}}} \left(\sup_{s_{p} \in W_{p}(M)\setminus\{0_{p}\}}
\frac{\vert s_{p} \vert^{'}_{p}}{\vert s_{p} \vert_{p}}\right)
=\sup_{p \in \overline{U_{i}}} \left(\sup_{S \in \mathbb{R}^{N}, |S|=1} \frac{\vert
  \sum\limits_{j=1}^{N}S_{j}e_{i,j} \vert^{'}(p)}{\vert
  \sum\limits_{j=1}^{N}S_{j}e_{i,j} \vert(p)}\right)
\end{eqnarray}
holds for any choice of $(U_{i},\varphi_{i})\in\mycal{A}$.  The right hand
side of (\ref{equivalenceequation})  is a finite positive number, because
it is nothing but  the maximum of a positive valued continuous function over
the compact  manifold $\overline{U_{i}}\times\mathbb{S}^{N-1}$, where
$\mathbb{S}^{N-1}$ denotes the $N-1$ dimensional unit sphere.  Let us denote
this positive number by $c_{i}$.  Then,
\begin{equation}
\label{equivalenceequation2}
\vert\cdot\vert^{'}\leq c_{i}\vert\cdot\vert
\end{equation}
holds over $\overline{U_{i}}$.

As an immediate consequence of (\ref{equivalenceequation2}) we have that
$\mathcal{F}_{i}\vert\cdot\vert^{'}\leq c_{i}\mathcal{F}_{i}\vert\cdot\vert$
holds throughout $M$, in accordance with the fact that $\mathcal{F}_{i}$ is
non-negative and $\supp(\mathcal{F}_{i})\subset U_{i}$.  Note{,} then, that the
sum $\sum_{i \in I}c_{i}\mathcal{F}_{i}$ is a positive valued $C^{r}$
function---which, as a consequence of the local finiteness of $\mycal{A}$, has
only finite non-zero terms in a sufficiently small neighborhood of any point
in $M$---, and that by definition $\sum_{i \in
  I}\mathcal{F}_{i}=1$. These, along with the above observations, implies then
that $\vert\cdot\vert^{'}\leq \left(\sum_{i \in
  I}c_{i}\mathcal{F}_{i}\right)\vert\cdot\vert$ holds, which justifies the
assertion of the lemma.
\end{proof}

We shall also apply the following two lemmas in verifying that the Sobolev
embedding theorem may also be applied in case of compact Riemann manifolds.

\begin{Lem}
\label{equivalence2}
Let $(\vert\cdot\vert_{l})_{l\in\{0,\dots,m\}}$ be norm fields, and $\nabla$,
$\nabla^{'}$ be   two $C^{r}$-covariant derivative operators. Then, there exists a
positive $C^{r}$-function $C$ over $M$ such that 
\begin{equation}
\label{equivalence2equation}
\sum_{l=0}^{r}\vert{\nabla^{\; '}}^{(l)}\cdot\vert_{l}\leq
C\sum\limits_{l=0}^{r}\vert\nabla^{(l)}\cdot\vert_{l}. 
\end{equation}
\end{Lem}
\begin{proof}
The  assertion of this lemma may be justified by combining
the following sequence of simple observations. 
\begin{enumerate}
\item the covariant 
derivation $\nabla^{'}$ can always be expressed as a sum of terms involving
$\nabla$ and the $C^{r-1}$ class Christoffel symbols,
\item the triangle inequality holds for norms,
\item the composition of a norm with a linear map is a semi-norm,
\item the sum of a norm and a semi-norm is a norm,
\item and, finally, by taking into account \emph{Lemma
  \ref{equivalence}}.
\end{enumerate}
\end{proof} 

Now, as a direct consequence of \emph{Lemmas} \ref{equivalence} and
\ref{equivalence2} we have the following.

\begin{Lem}
\label{equivalence3}
Let $(\vert\cdot\vert_{l})_{l\in\{0,\dots,m\}}$ and
$(\vert\cdot\vert^{'}_{l})_{l\in\{0,\dots,m\}}$ be  norm field collections,
furthermore, $\nabla$ and $\nabla^{'}$ be two covariant derivative operators
as above. Then, there exists a positive $C^{r}$ function $C$ on $M$
such that 
\begin{equation}
\label{equivalence3equation}
\sum_{l=0}^{r}\vert{\nabla^{\; '}}^{(l)}\cdot\vert^{\; '}_{l}\leq
C\sum_{l=0}^{r}\vert\nabla^{(l)}\cdot\vert_{l}. 
\end{equation}
\end{Lem}

Being armed with the above results we can turn to the generalization of
Sobolev embedding theorem to compact  Riemann manifolds.

\begin{Thm}
\label{sobolevembeddingtheorem}
The Sobolev embedding theorem applies to compact Riemann manifolds, i.e.{,} 
whenever $M$ is a compact Riemann manifold of dimension $n$ and 
$r>\frac{n}{2}+k$ for non-negative integers $r$ and $k$, the relation
\begin{equation}
\label{sobolevinclusionC}
H_{r}^{2}(M,\C)\subset C^{k}_{\infty}(M,\C)
\end{equation}
holds, and also there exists a positive real constant $C$ such that the  
inequality
\begin{equation}
\label{sobolevinequalityC}
\Vert\cdot\Vert_{C^{k}_{\infty}(M,\C)}\leq C\,\Vert\cdot\Vert_{H_{r}^{2}(M,\C)}
\end{equation}
is satisfied.
\end{Thm}
\begin{proof}
Obviously, the inclusion (\ref{sobolevinclusionC}) holds as an immediate
consequence of the relevant differentiability assumptions, and because 
$C^{k}_{\infty}(M,\C)=C^{k}_{b}(M,\C)=C^{k}(M,\C)$ due to the compactness of 
$M$.

To justify (\ref{sobolevinequalityC}) choose an arbitrary finite atlas
$\mycal{A}=\{(U_{i},\varphi_{i}) \vert i \in I \}$ of
$M$ with partition of unity {$\{\mathcal{F}_{i} \vert i\in I\}$} 
subordinate to it. In proceeding, choose a real  number $\varepsilon$ such 
that $0<\varepsilon\leq1$, and denote by $V_{i}^{\varepsilon}$ the 
pre-image  of the interval $[\varepsilon,1]$ by the map $\mathcal{F}_{i}$ 
for each $i\in I$. Then, by applying the conventional Sobolev theorem over each compact set 
$V_{i}^{\varepsilon}\subset\supp(\mathcal{F}_{i})\subset U_{i}$ we have that 
\begin{equation}
\label{sobolevembeddingtheoremequation1}
\sup_{V_{i}^{\varepsilon}} \left(\sum_{l=0}^{k}\vert\nabla^{(l)}f\vert\right)^{2} 
\leq C_{i}\, \sum_{l=0}^{r}\int_{V_{i}^{\varepsilon}}\vert\nabla^{(l)}f\vert^{2} 
\end{equation}
for any $f\in H_{r}^{2}(M,\C)$, where now $\vert\cdot\vert$, $\nabla$ and
the volume form are assumed to {be} determined by the Euclidean metric  
associated with the local coordinates, $\varphi_{i}(U_{i})\subset \R^n$, on $U_{i}$.
In virtue of \emph{Lemma
  \ref{equivalence3}}, one may  replace the Euclidean $\vert\cdot\vert$,
$\nabla$ and the volume form with the norm, covariant derivation and the
volume form determined by the Riemannian metric $g$ on $M$ with the
understanding that, as a compensation, the values of Sobolev constants $C_{i}$
have to be adjusted accordingly. Furthermore, as 
$\mathcal{F}_{i}\geq\varepsilon$ over $V_{i}^{\varepsilon}$, we also have that 
\begin{equation}
\label{sobolevembeddingtheoremequation2}
\sup_{V_{i}^{\varepsilon}}
\left(\sum_{l=0}^{k}\vert\nabla^{(l)}f\vert\right)^{2} \leq
C_{i}\frac{1}{\varepsilon}\,
\sum_{l=0}^{r}\int_{V_{i}^{\varepsilon}}\mathcal{F}_{i}\vert\nabla^{(l)}f\vert^{2}\,,
\end{equation}
where the monotonicity of the integral of a non-negative function has been taken into account. By making
use of the above observations we also have that 
\begin{eqnarray}
\label{sobolevembeddingtheoremequation3}
&&\hskip-1.3cm\sum_{l=0}^{r}\int_{M}\vert\nabla^{(l)}f\vert^{2} = \sum_{i\in I}
\sum_{l=0}^{r}\int_{U_{i}}\mathcal{F}_{i}\vert\nabla^{(l)}f\vert^{2}\cr 
&&\hskip1.cm\geq \sum_{i\in I}
\sum_{l=0}^{r}\int_{V_{i}^{\varepsilon}}\mathcal{F}_{i}\vert\nabla^{(l)}f\vert^{2}\cr 
&&\hskip1.0cm\geq \sum_{i\in I} \frac{\varepsilon}{C_{i}}\sup_{V_{i}^{\varepsilon}}
\left(\sum_{l=0}^{k}\vert\nabla^{(l)}f\vert\right)^{2}. 
\end{eqnarray}
This is exactly the point where we {utilize} the compactness of
$M$. Accordingly, for the rest of the proof we shall assume that the index set
$I$ is finite{,} which immediately implies that for some positive constant
$C^{\varepsilon}$
\begin{eqnarray}
\label{sobolevembeddingtheoremequation4}
\sum_{i\in I} \frac{\varepsilon}{C_{i}} \sup_{V_{i}^{\varepsilon}}
\left(\sum_{l=0}^{k}\vert\nabla^{(l)}f\vert\right)^{2}
\geq C^{\varepsilon} \sup_{{{\bigcup} \atop {i\in I}}V_{i}^{\varepsilon}}
\left(\sum_{l=0}^{k}\vert\nabla^{(l)}f\vert\right)^{2} 
\end{eqnarray}
holds. By choosing $\varepsilon$ such that 
$0<\varepsilon\,
{\leq}\,
\frac{1}{\vert I\vert}$ 
we also have that ${{\bigcup} \atop {i\in I}}V_{i}^{\varepsilon}=M$
since otherwise there would be a point $p\in M$ so that
$\mathcal{F}_{i}(p)<\varepsilon$ for all $i\in I$. This, however, is
impossible since then 
$\sum_{i\in I}\mathcal{F}_{i}(p)<\sum_{i\in I}\varepsilon
{\leq} 
1$
would hold on contrary to the definition of the
partition of unity. This, in virtue of 
(\ref{sobolevembeddingtheoremequation3}) and
(\ref{sobolevembeddingtheoremequation4}), justify then that
(\ref{sobolevinequalityC}) holds. 
\end{proof}

In applying the above results consider now a compact Riemann manifold
$M$. As it is well-known \cite{laplacecomp}
the eigenfunctions of the Laplace operator are in $C^{r}_{b}(M,\C)$,
and{,} more importantly, their linear span is dense in any of the
Banach spaces $C^{l}_{b}(M,\C)$, with $0\leq l\leq r$. In
addition, $C^{l}_{b}(M,\C)\subset H_{l}^{2}(M,\C)$ is dense and---in consequence of the H\"older's inequality---the $C^{l}$
supremum norm is stronger than the $H_{l}^{2}$ norm over compact
manifolds. 
Therefore, the linear span of the Laplace eigenfunctions 
is dense in $H_{l}^{2}(M,\C)$, and, in particular, in $L^2(M,\C)$. 
It is also known that a linearly independent eigensystem of the Laplace
operator over a compact manifold may be chosen to be orthonormal with
respect to the $L^{2}$ scalar product, hence these form a complete
orthonormal system in $L^{2}(M,\C)$, a complete orthogonal system in
$H_{l}^{2}(M,\C)$, as well as, a Schauder basis in $C^{l}_{b}(M,\C)$,
with $0\leq l\leq r$.  Let $f\in H_{r}^{2}(M,\C)$ be some function and
$\{Y_{i}\}_{i\in \mathcal{I}}$  {be an  $L^{2}$-orthonormal}
eigensystem of the Laplace operator. Then, as the elements of  the
system $\{Y_{i}\}_{i\in \mathcal{I}}$ are eigenvectors of the Laplace
operator, in virtue of the  Gauss theorem, we have that
\begin{equation}
\label{sobolevprojection}
\left<Y_{i},f\right>_{H_{r}^{2}(M,\C)}=\left(\sum_{l=0}^{r}(-\lambda_{i})^{l}\right)
\left<Y_{i},f\right>_{L^{2}(M,\C)} 
\end{equation}
for each $i\in \mathcal{I}$, where $\lambda_{i}$ denotes the associated
eigenvalue, and $\partial M=\emptyset$ has also been assumed. In order to
simplify some of the succeeding expressions we introduce the function
$S_{r}:\C\rightarrow \C$ as
\begin{equation} 
S_{r}(z)=\sum_{l=0}^{r}(-z)^{l}= \left\{ \begin{array} {r l}
  \frac{(-z)^{r+1}-1}{(-z)-1} , & 
{\rm if}\  z\neq -1;\cr r+1 , & {\rm otherwise}. \end{array}
\right.\label{ooo}
\end{equation}
As an immediate consequence of (\ref{sobolevprojection})  the vector system 
$\{Y_{i}/\sqrt{S_{r}(\lambda_{i})}\}_{i\in \mathcal{\mathcal{I}}}$  is orthonormal in
$H_{r}^{2}(M,\C)$.  Therefore, for the series expansion of the function $f$ in
$H_{r}^{2}(M,\C)$ with respect  to the complete orthonormal system
$\{Y_{i}/\sqrt{S_{r}(\lambda_{i})}\}_{i\in \mathcal{\mathcal{I}}}$ 
\begin{eqnarray}
\label{sobolevexpansion}
\sum_{i\in  \mathcal{I}}
\left<\frac{Y_{i}}{\sqrt{S_{r}(\lambda_{i})}},f\right>_{{H_{r}^{2}(M,\C)}} 
\cdot\frac{Y_{i}}{\sqrt{S_{r}(\lambda_{i})}} 
=\sum_{i\in \mathcal{I}}\left<Y_{i},f\right>_{L^{2}(M,\C)}Y_{i},
\end{eqnarray}
is satisfied, which relation implies then that the $L^2$ series expansion of
$f$ with respect to $\{Y_{i}\}_{i\in \mathcal{I}}$  is also convergent in the
$H_{r}^{2}$ sense. 

By combining the above observations we have that whenever the assumptions of
the Sobolev embedding theorem holds then the multipole series
expansions---which are convergent in the $L^{2}$ sense---will also be
convergent in the uniform $C^{k}$ sense. 

\begin{Cor}
\label{rultipoleexpansion}
Assume that $f$ is a function that belongs to $C^{r}(M,\C)$, where
$r>\frac{n}{2}+k$. Then, the multipole series expansion of $f$ with respect to an 
{eigensystem} $\{Y_{i}\}_{i\in \mathcal{I}}$ {of the Laplace operator} 
is also
convergent in the $C^{k}$ sense.  Furthermore, the convergence is independent
of the summation order, as the  system $\{Y_{i}\}_{i\in \mathcal{I}}$  shall
remain to be a complete orthogonal  system in $L^{2}(M,\C)$ and in
$H_{r}^{2}(M,\C)$ after any index permutation.  Therefore, pointwise absolute
convergence of the multipole series also follows  for the derivatives up to
the order $k$, as the absolute convergence in a finite  dimensional Banach
space (which is nothing but $\C$ in the present case) is equivalent to
summation order independent convergence.
\end{Cor}

\begin{Rem}
\label{rultipoleexpansionremark}
It is worth keeping in mind that the Schauder basis property of the Laplace
eigenfunctions in $C^{k}_{b}(M,\mathbb{C})$ does not guarantee that an arbitrary function $f\in C^{k}_{b}(M,\mathbb{C})$ can always be expanded in the
form of $f=\sum_{i\in\mathcal{I}}f_{i}Y_{i}$. (If that was true the relation $f_{i}=\left<Y_{i},f\right>_{L^{2}}$, $i\in\mathcal{I}$, would immediately follow from Lebesgue's theorem of dominated convergence.) Our statement follows from the fact that whenever $\{Y_{i}\vert i\in\mathcal{I}\}$ is a Schauder basis that guarantees merely that the linear span of this system is dense in $C^{k}_{b}(M,\mathbb{C})$, i.e.\ every field $f$ may be approximated by a sequence of finite linear combinations of this basis. This means that the series expandability with respect to a Schauder basis 
follows automatically only in Hilbert spaces, but not in general in Banach spaces.
\end{Rem}

It is important to emphasize that the pointwise absolute convergence property plays a crucial role in non-linear problems. To see this
note that whenever the conditions of the above corollary are
guaranteed to hold the pointwise product $f\cdot g$ of multipole
expansions $f=\sum_{i\in \mathcal{I}}f_{i}Y_{i}$ and $g=\sum_{i\in
\mathcal{I}}g_{i}Y_{i}$ may be written in the optimal form
\begin{eqnarray}
\label{rultilinearexpansion}
\hskip-.4cm f\cdot g=\left(\sum_{i\in
  \mathcal{I}}f_{i}Y_{i}\right)\left(\sum_{j\in
  \mathcal{I}}g_{j}Y_{j}\right) =  
\sum_{i,j\in \mathcal{I}}f_{i}\,g_{j}\cdot Y_{i}Y_{j}\,{,}
\end{eqnarray}
where the sums are understood in the pointwise manner, which, due to
Fubini's theorem, are absolute convergent along with their derivatives
up to the order $k$. The sums may also be understood in  the uniform
$C^{k}$ manner independently of the summation order, however, in  this
case we do not get \emph{absolute} uniform $C^{k}$ convergence.

The most significant advantages associated with the use of multipole analysis
manifest themselves in evaluating non-linear terms. In fact, whenever the
pointwise 
absolute convergence is guaranteed for $f\cdot g$ we have that
\begin{eqnarray}
f\cdot g=\sum_{k\in   \mathcal{I}} \left(f\cdot g\right)_k Y_{k}\,,
\end{eqnarray} 
where
\begin{eqnarray}
\label{rultilinearexpansion2}
&&\hskip-.4cm \left(f\cdot g\right)_k=\left<Y_{k},f\cdot g\right>\cr
&&\hskip-.4cm \phantom{\left(f\cdot g\right)_k}=\int_M \bar{Y}_{k}
\left(\sum_{i\in \mathcal{I}}f_{i}Y_{i}\right)\left(\sum_{j\in 
  \mathcal{I}}g_{j}Y_{j}\right)\cr 
&&\hskip-.4cm \phantom{\left(f\cdot g\right)_k}=\sum_{i\in \mathcal{I}}
\sum_{j\in \mathcal{I}}f_{i}g_{j}\int_M \bar{Y}_{k}Y_{i}Y_{j}\,. 
\end{eqnarray}
In deriving  (\ref{rultilinearexpansion2}) we have used that, in virtue of
Lebesgue's dominated convergence theorem, the summation and the integration
can be interchanged, and also that the double sums are known to be uniformly
convergent. Therefore, whenever the matrix elements $\int Y_{i}Y_{j}Y_{k}$ with
$i,j,k\in \mathcal{I}$, which are also known as Gaunt coefficients, can be
determined{,} then, the multipole series expansion on multi-linear expressions may
constructively  be evaluated by making use of the multipole expansion of its
factors.

\section{The Explicit Form of Gaunt Coefficients on Two-sphere and on $n$-dimensional Torus}
\label{gauntcoefficients}

It is known (see, e.g., \cite{gaunt,gaunt2}) that the Gaunt coefficients on a 
two-sphere, $\mathbb{S}^2$ , may be given either in terms of the Wigner 3j symbols or in terms
of the Clebsch-Gordan  coefficients. It is also known that the Clebsch-Gordan
coefficients can be evaluated by various numerical algorithms, e.g.{,} by the
Racah formula \cite{ressiah} (see also \cite{gsl}). 

In providing the Gaunt coefficients let us denote the spherical harmonics, which comprises the familiar orthonormal basis, by $Y{}_\ell^m$ with integers $\ell,m$ satisfying the relations {$0\leq \ell <\infty$} and $-\ell\leq m\leq \ell$. 

Then, the Gaunt coefficient $\int_{{\mathbb{S}}^2}
Y{}_{\ell_{1}}^{m_{1}} Y{}_{\ell_{2}}^{m_{2}} Y{}_{\ell_{3}}^{m_{3}}$ is known to be zero if
either of the following conditions holds
\begin{itemize}
 \item[(i)] $m_{1}+m_{2}+m_{3}\neq0$, or 
 \item[(ii)] $\ell_{3}<\vert \ell_{1}-\ell_{2}\vert$ or $\ell_{1}+\ell_{2}<\ell_{3}$, or 
 \item[(iii)] $\ell_{1}+\ell_{2}+\ell_{3}$ is odd.
\end{itemize}
The Gaunt coefficients may also be related to the Wigner 3j symbols as
\begin{eqnarray}
\label{gauntformula}
&&\hskip-1.4cm\int_{S^2} Y{}_{\ell_{1}}^{m_{1}}Y{}_{\ell_{2}}^{m_{2}}Y{}_{\ell_{3}}^{m_{3}}\cr
&&\hskip-.4cm=\sqrt{\frac{(2\ell_{1}+1)(2\ell_{2}+1)(2\ell_{3}+1)}{4\pi}}
\cdot\left({\ell_{1}\atop 0}{\ell_{2}\atop 0}{\ell_{3}\atop
  0}\right)\left({\ell_{1}\atop m_{1}}{\ell_{2}\atop m_{2}}{\ell_{3}\atop
  m_{3}}\right), 
\end{eqnarray}
where the Wigner 3j symbols are given by the Clebsch-Gordan coefficients as
\begin{eqnarray}
\label{wigner3jsymbol}
\left({\ell_{1}\atop m_{1}}{\ell_{2}\atop m_{2}} {\ell_{3}\atop
  m_{3}}\right) = 
\frac{(-1)^{-m_{3}+\ell_{1}-\ell_{2}}}{\sqrt{2\ell_{3}+1}}
\cdot(\ell_{1}\,\ell_{2}\,m_{1}\,m_{2}\,\vert\,\ell_{1}\,\ell_{2}\,\ell_{3}\,(-m_{3})).
\end{eqnarray}

The Gaunt coefficients relevant for the case of an $n$-dimensional torus{,} $\T^n${,} are known to possess an even {simpler} 
structure. By making use of
the orthonormal eigenstates $Y_{k^{1},\dots,k^{n}}$, labeled by the integers
$k^{1},\dots,k^{n}$ of the Laplace operator on $\T^n${,} the integrals
$\int_{\T^n} Y_{k_{1}^{1},\dots,k_{1}^{n}}\cdot
Y_{k_{2}^{1},\dots,k_{2}^{n}}\cdot Y_{k_{3}^{1},\dots,k_{3}^{n}}$ may be seen
to be  zero if there exists $i\in\{1,\dots,n\}$ such that
$k_{1}^{i}+k_{2}^{i}+k_{3}^{i}\not=0$ while it is
$\frac{1}{(\sqrt{2\pi})^{n}}$ otherwise.

\section{The Division as an Operation within the Framework of Multipole
  Expansions} 
\label{divisionbyfield}

One of the most delicate issues while using multipole expansions arises whenever we have expressions containing the operation of division by a
function. Clearly, if this function vanishes somewhere, then, a direct evaluation of the division is not tolerated by any of
the numerical methods. However, if this function is guaranteed to be
bounded and nowhere zero, then both the multiplication and the division by it
are continuous operations in either of the function spaces $L^{2}$,
$C^{k}_{b}$ and $H^{2}_{k}$ with the understanding that in the latter two
cases  the result belongs to the corresponding spaces if the function itself
belongs to $C^{k}_{b}$. In the latter case, by making use of the Neumann series
expansion  the action of the division operator may be traced back to the
multiple use of   the multiplication operator.

To justify our last claim consider first a continuous operator $A$ acting on a Banach space with identity operator $I$. It is straightforward to check then{,} by induction, that
\begin{equation}
\label{neumannformula}
A\sum_{i=0}^{N}(I-A)^{i}=\sum_{i=0}^{N}(I-A)^{i}A=I-(I-A)^{N+1}
\end{equation}
{holds} for arbitrary non-negative integer $N$. If $\Vert I-A\Vert<1$
we also have that the series $N\mapsto\sum_{i=0}^{N}(I-A)^{i}$  is
absolute convergent (thus its limit is a continuous operator), and
$\Vert(I-A)^{N+1}\Vert\leq\Vert I-A\Vert^{N+1}$ tends to zero  as
$N\rightarrow \infty$. Thus the relation $A^{-1}=\sum_{i=0}^{\infty}(I-A)^{i}$---referred as the Neumann series expansion---follows.  Assume now that
the Banach space in question is $C^{k}_{b}(M,\C)$, and the operator
$A$ is the multiplication by a function $F\in C^{k}_{b}(M,\C)$, and
denote by $\Vert I-F\Vert$ the pertinent operator norm of the function
$1-F$. Accordingly, an upper bound for the error that we introduce by
replacing  the division operation by multiplication based on the
Neumann series may be given as
\begin{eqnarray}
\label{neumanndivision}
\left\Vert 1-F\sum_{i=0}^{N}(1-F)^{i}\right\Vert=\left\Vert
1-\sum_{i=0}^{N}(1-F)^{i}F\right\Vert\cr 
\leq \Vert 1-F\Vert_{C^{k}_{b}}^{N+1}.
\end{eqnarray}
When $F$ is specified by its multipole coefficients, it is not economical to determine the $C^{k}_{b}$ norm of $1-F$, as it would require pointwise
evaluation of its multipole  series. However, if $F\in C^{r}_{b}$ for $r>\frac{n}{2}+k$, and our Riemann  manifold is compact, the Sobolev embedding
theorem significantly simplifies the determination of an upper bound of the uniform $C^{k}$ norm as 
\begin{equation}
\label{divisionerror}
\Vert 1-F\Vert_{C^{k}_{b}} = \Vert 1-F\Vert_{C^{k}_{\infty}}\leq C\, \Vert
1-F\Vert_{H^{2}_{r}}, 
\end{equation}
where $C$ is the minimal Sobolev constant, and the Sobolev norm on the right hand side may  directly be determined by making use of the multipole coefficients of $F$ as
\begin{equation}
\label{divisionerrorsobolev}
\sqrt{\sum_{i\in \mathcal{I}}S_{r}(\lambda_{i})\left\vert\left<Y_{i},1\right>_{L^{2}}
  -\left<Y_{i},F\right>_{L^{2}}\right\vert^{2}} 
\end{equation}
where the function $S_r$ introduced in \ref{sobolevembedding} has
been applied. Given the value of $C\,\Vert 1-F\Vert_{H^{2}_{r}}<1$, a
predefined error tolerance $\varepsilon$ can be guaranteed to hold simply by
calculating Neumann series up to {the} order
$N_{\varepsilon}=\mathrm{int}\left\{\ln(\varepsilon)/\ln(C\,\Vert
1-F\Vert_{H^{2}_{r}})\right\}$, where $\mathrm{int}\left\{x\right\}$ denotes
the  integer part of $x\in\R$. It is straightforward to see that the number of
orders to be  calculated grows only logarithmically with the increase of the
desired accuracy.

The above described method based on the use of the Neumann series may further
be optimized by rescaling our field $F$  with a complex number $z$ in a way to
minimize $\left(\Vert 1-z F\Vert_{H^{2}_{r}}\right)^{2}$. It can be justified, by a direct
calculation, that the unique minimum may be achieved by choosing 
\begin{equation}
\label{rinimumat}
z(F)=\frac{\sum_{i\in
    I}S_{r}(\lambda_{i})\overline{\left<Y_{i},F\right>}_{L^{2}}
\left<Y_{i},1\right>_{L^{2}}}{\sum_{j\in \mathcal{I}}S_{r}(\lambda_{j})\left\vert
\left<Y_{j},F\right>_{L^{2}}\right\vert^{2}},
\end{equation}
and the corresponding minimal value is
\begin{eqnarray}
\label{rinimalvalue}
M(F)=\sum_{i\in
  I}S_{r}(\lambda_{i})\left\vert\left<Y_{i},1\right>_{L^{2}} 
\right\vert^{2}
-\frac{\left\vert\sum_{j\in \mathcal{I}}S_{r}(\lambda_{j})
\overline{\left<Y_{j},F\right>}_{L^{2}}\left<Y_{j},1\right>_{L^{2}}
\right\vert^{2}}{\sum_{k\in \mathcal{I}}S_{r}(\lambda_{k})\left\vert
\left<Y_{k},F\right>_{L^{2}}\right\vert^{2}}.
\end{eqnarray}
If we denote the constant eigenfunction of the Laplace operator by $Y_{0}$,
then $z(F)$ and $M(F)$ can be re-expressed as
$z(F)=\frac{1}{Y_{0}}\frac{\overline{\left<Y_{0},F\right>}_{L^{2}}}
{\left\Vert F\right\Vert_{H^{2}_{r}}^{2}}$ and   $M(F)=\frac{1}{\vert
  Y_{0}\vert^{2}}\frac{\Vert
  F-\left<Y_{0},F\right>_{L^{2}}Y_{0}\Vert_{H^{2}_{r}}^{2}}{\Vert
  F\Vert_{H^{2}_{r}}^{2}}$.   If optimal rescaling is applied, {a} sufficient
condition for the Neumann  series to converge is {that} the inequality
$C\cdot\sqrt{M(F)}<1$, where $C$ stands for the Sobolev constant, {holds} and
the minimum number of orders necessary to achieve an accuracy below a
pre-fixed value $\varepsilon$ is
$N_{\varepsilon}=\mathrm{int}\{\ln(\varepsilon)/\ln(C\,\sqrt{M(F)})\}$.  This
requirement may also be rephrased as follows. The Neumann series after optimal
rescaling is absolute convergent in the uniform $C^{k}$ norm if $\frac{\Vert
  F-\left<Y_{0},F\right>_{L^{2}} Y_{0}\Vert_{H^{2}_{r}}}{\Vert
  F\Vert_{H^{2}_{r}}}$---which is nothing but the $H^{2}_{r}$-measure of the
non-monopole content in $F$---is   smaller than the threshold
$\frac{\vert Y_{0}\vert}{C}$. As it will be demonstrated by the following two
examples relevant for {the} case of {the} two-sphere{s} and $n$-dimensional toruses, this is
a rather weak condition.  It is also worth to note that the Neumann series
expansion may be re-expressed  in an iterative form, which requires less
function evaluations than  the canonical series expansion representation does
\cite{neumann}.

\section{The Sobolev Constants on Two-sphere and on $n$-Toruses}
\label{sobolevconstant}

As it follows from the discussions in \ref{divisionbyfield}, to be able to have accurate estimates of certain errors, it is also important to know the numerical value of the minimal
Sobolev constant. In general, the determination of the value of the  minimal
Sobolev constant is a delicate issue. A powerful method yielding this
constant is based on the use of \emph{reproducing kernel  property}
\cite{reprkernel} of Sobolev spaces which can be outlined as
follows. Consider a Hilbert space $\mycal{H}$ of some complex valued functions
over some set $X$. Then, if the point evaluation $f\mapsto f(x)$ is a continuous linear
map for every $x\in X$ it is called a reproducing kernel Hilbert
space. The Riesz representation theorem ensures that for each $x\in X$ there exists a unique $K_{x}\in \mycal{H}$ such that $\left<K_{x},f\right>=f(x)$ for any $f\in \mycal{H}$. As $K_{x}$ itself is a function, it may also be evaluated at any point. The reproducing kernel function $K:X\times X\mapsto\C$ is defined as $K(x,y)=K_{x}(y)${.} {I}t may be verified
that for any $x,y\in X$
\begin{itemize}
\item[(i)] $\left<K(x,\cdot),K(y,\cdot)\right>=K(y,x)$,  
\item[(ii)] $\bar{K}(y,x)=K(x,y)$,
\item[(iii)] if $\left(\Phi_{i}\right)_{i\in \mathcal{I}}$ comprises a
  complete  orthonormal system, then{,} $K(x,\cdot)=\sum_{i\in
  \mathcal{I}} \bar{\Phi}_{i}(x)\Phi_{i}$, where the infinite summation makes sense in the norm topology.
\end{itemize}
If $H^{2}_{r}(M,\C)$ is a Sobolev space over a compact Riemann
manifold $M$ with $r>\frac{n}{2}+k$, then by the Sobolev theorem the
$H^{2}_{r}$ norm is stronger than the  uniform $C^{k}$ norm, therefore
the point evaluation is a continuous map. Thus, for an arbitrary
choice of {an  $L^{2}$-orthonormal eigensystem} $\{Y_i\}_{i\in
\mathcal{I}}$ {of the Laplace operator,}  there exists a unique
reproducing kernel function $K_{r}$ which---for any choice of $x\in M$ reads---as $K_{r}(x,\cdot)=\sum_{i\in
\mathcal{I}}\frac{1}{S_{r}(\lambda_{i})}\bar{Y}_{i}(x)Y_{i}$, and for any $f\in H^{2}_{r}(M,\C)$ we have
$\left<K_{r}(x,\cdot),f\right>_{H^{2}_{r}(M,\C)}=f(x)$.      Armed
with this identity, the relation
\begin{equation}
\label{repgradient}
\left<\left(\left(\nabla^{(l)}\otimes \id\right)K_{r}\right)
(x,\cdot),f\right>_{H^{2}_{r}(M,\C)}=\left(\nabla^{(l)}f\right)(x)
\end{equation}
{can be seen to hold,} where the operator $\left(\nabla^{(l)}\otimes
\id\right)$ acts on $K_{r}$ as the  $l$-times gradient---with $l=0,\dots,k$---on the first variable of $K_{r}$ while the second variable of $K_{r}$
remains intact.  Note that, in virtue of Lebesgue's dominated
convergence theorem, the order of the action of the gradient $\nabla$
and the scalar product {$\left<\cdot,\cdot\right>$} may be interchanged. {All these} observations imply that
\begin{eqnarray}
\label{repgradientabs}
&&\hskip-1.1cm\left\vert\nabla^{(l)}f\right\vert(x)\cr
&&\hskip-.7cm=\left\vert\left<\left(\left(\nabla^{(l)}\otimes \id\right)
K_{r}\right)(x,\cdot),f\right>_{H^{2}_{r}(M,\C)}\right\vert\cr 
&&\hskip-.7cm\leq\left\Vert\left(\left(\nabla^{(l)}\otimes
\id\right)K_{r}\right) (x,\cdot)\right\Vert_{H^{2}_{r}(M,\C)} \, \left\Vert
f\right\Vert_{H^{2}_{r}(M,\C)}\hskip-.1cm, 
\end{eqnarray}
where in the second step the Cauchy-Schwartz inequality in
$H^{2}_{r}(M,\C)$ has been applied. This inequality is known to be sharp,
i.e.\,it may be saturated. When the norm
$\left\Vert\left(\left(\nabla^{(l)}\otimes
\id\right)K_{r}\right)(x,\cdot)\right\Vert_{H^{2}_{r}(M,\C)}$  is guaranteed to be independent of $x\in M$---this happens{,} e.g.\,in case of homogeneous manifolds---we get the sharp inequality
\begin{eqnarray}
\label{supgradient}
 \sup_{x\in M} \left\vert\nabla^{(l)}f\right\vert(x)
\leq\sup_{x\in M} \left\Vert\left(\left(\nabla^{(l)}\otimes \id\right)
K_{r}\right)(x,\cdot)\right\Vert_{H^{2}_{r}(M,\C)}
\,\left\Vert f\right\Vert_{H^{2}_{r}(M,\C)}.
\end{eqnarray}
Applying this relation to the case of $k=0$ the minimal Sobolev constant may be read off the particular form of (\ref{supgradient}) as 
\begin{equation}
\label{rinimalsobolevconstant}
\left\Vert f\right\Vert_{C^{0}}\leq \sup_{x\in M} \left\Vert
K_{r}(x,\cdot)\right\Vert_{H^{2}_{r}(M,\C)} \left\Vert
f\right\Vert_{H^{2}_{r}(M,\C)}\,, 
\end{equation}
where $r>\frac{n}{2}$ is tacitly assumed to hold. Note, however, that for higher value of $k$ the above argument does not necessarily lead to a sharp inequality.

\smallskip

Let us restrict again considerations to the case of a two-sphere where, according to the above discussion, the reproducing kernel can be given as
\begin{equation}
\label{s2sobolevconstant}
K_{r}(x,\cdot)=\sum_{\ell=0}^{\infty}\sum_{m=-\ell}^{\ell}
\frac{\ell(\ell+1)-1}{(\ell(\ell+1))^{r+1}-1}\bar{Y}_{\ell}^{m}(x)Y{}_{\ell}^{m} \,.
\end{equation}
Then, as the two-sphere is a homogeneous manifold
$\left\Vert\left(\left(\nabla^{(l)}\otimes
\id\right)K_{r}\right)(x,\cdot)\right\Vert_{H^{2}_{r}(M,\C)}$   is independent
of the location of $x$ on $\mathbb{S}^{2}$. By choosing $x$ to be the north
pole in standard spherical polar coordinates,  and also by using the values of
the spherical harmonics at the north pole we immediately get that 
\begin{equation}
\label{s2sobolevconstant2}
\left\Vert K_{r}(x,\cdot)\right\Vert_{H^{2}_{r}}^{2} =
\frac{1}{4\pi}\sum_{\ell=0}^{\infty}\frac{(2\ell+1)(\ell(\ell+1)-1)}{(\ell(\ell+1))^{r+1}-1}. 
\end{equation}
The square root of the right hand side of (\ref{s2sobolevconstant2}) provides
the minimal value of the Sobolev constant, $C_r$, over the two-sphere with
$k=0$ and $r>1$. The approximate numerical values of this Sobolev constants
$C_r$ for the particular values of $r=2,3,4$ are listed in Table~\ref{sobvaluess2}.

\begin{table}[!ht]
\begin{center}
\begin{tabular}{l|l}
$r$ & Sobolev constant \\
\hline
$2$   & $\frac{1}{\sqrt{4\pi}}\cdot1.284533$ \\
$3$   & $\frac{1}{\sqrt{4\pi}}\cdot1.106732$ \\
$4$   & $\frac{1}{\sqrt{4\pi}}\cdot1.048986$ \\
\end{tabular}
\caption{\label{sobvaluess2} Approximate values of the minimal Sobolev constant in 
the $C^{0}\subset H_{r}^{2}$ Sobolev embedding over the two-sphere, for the
$r=2,3,4$ values.}
\end{center}
\end{table}

Let us finally restrict attention to the case of an $n$-dimensional torus. Then, the reproducing kernel may be given as 
\begin{eqnarray}
\label{tnsobolevconstant}
K_{r}(x,\cdot)
=\sum_{k_{1}=-\infty}^{+\infty}\dots\sum_{k_{n}=-\infty}^{+\infty}
\frac{\bar{Y}_{k_{1},\dots,k_{n}}(x)\cdot Y_{k_{1},\dots,k_{n}}}
     {S_{r}(-k_{1}^{2}-\dots-k_{n}^{2})}, 
\end{eqnarray}
where $\left\Vert\left(\left(\nabla^{(l)}\otimes
\id\right)K_{r}\right)(x,\cdot)\right\Vert_{H^{2}_{r}}$   is constant as a
function of $x\in \mathbb{T}^{n}$, as the $n$-torus is also a  homogeneous
manifold. By choosing $x$ to be  the point where all the polar angle
coordinates are zero we immediately get 
\begin{eqnarray}
\label{tnsobolevconstant2}
\left\Vert K_{r}(x,\cdot)\right\Vert_{H^{2}_{r}}^{2}
=\frac{1}{(2\pi)^{n}}
\sum_{k_{1}=-\infty}^{+\infty}\dots\sum_{k_{n}=-\infty}^{+\infty} 
\frac{1}{S_{r}(-k_{1}^{2}-\dots-k_{n}^{2})}.
\end{eqnarray}
The square root of this expression gives the minimal Sobolev constants, $C_r$,
for the Sobolev inequality with $r>\frac{n}{2}$ and $k=0$. For the particular
case of $\T^2$, the approximate numerical   value of the Sobolev constants
$C_r$ for $k=0$ and for the particular values of $r=2,3,4$ are listed in 
Table~\ref{sobvaluest2}.

\begin{table}[!ht]
\begin{center}
\begin{tabular}{l|l}
$r$ & Sobolev constant \\
\hline
$2$   & $\frac{1}{2\pi}\cdot1.943685$ \\
$3$   & $\frac{1}{2\pi}\cdot1.547391$ \\
$4$   & $\frac{1}{2\pi}\cdot1.397749$ \\
\end{tabular}
\caption{\label{sobvaluest2} Approximate values of the minimal Sobolev constant in 
the $C^{0}\subset H_{r}^{2}$ Sobolev embedding over the two-torus, for the  
$r=2,3,4$ values.}
\end{center}
\end{table}

\section{The Estimation of the Tail Sum Error}
\label{tailsumerrorbounds}

As our numerical method is based on the use of multipole expansion of the
basic field variables, and also since, in practice, we always use
only a finite number of multipole components, it is of crucial importance to
provide precise estimates on the pertinent errors. An immediate 
upper bound on the Sobolev norm of the truncated part of a function $f$ may be given as follows. 

Let $f\in H^{2}_{r}(M,\C)$ for some $r$, then, in virtue of
(\ref{sobolevprojection}), we have that 
\begin{equation}
\label{sobolevnormspectral}
\sum_{i\in \mathcal{I}}S_{r}(\lambda_{i})\left\vert\left<Y_{i},f\right>_{L^{2}}
\right\vert^{2} = \Vert f\Vert_{H^{2}_{r}}^{2}<\infty\,,
\end{equation}
which implies that the sequence 
$i\mapsto
S_{r}(\lambda_{i})\left\vert\left<Y_{i},f\right>_{L^{2}}\right\vert^{2}$
is summable. Consider now a sequence of positive numbers 
$i\mapsto a_{i}$ such that the relation 
\begin{equation}
\label{sobolevnormbound}
\left\vert\left<Y_{i},f\right>_{L^{2}}\right\vert \leq a_{i}
\end{equation}
holds for each $i\in \mathcal{I}$. Then, because the summation
preserves monotonicity, the tail sum of the sequence 
$i\mapsto S_{r}(\lambda_{i})\left\vert a_{i}\right\vert^{2}$ bounds the 
$H_{r}^{2}$ norm-square of the tail sum error of $f$. Such a bounding
sequence may be readily constructed by assuming that  
$f\in H^{2}_{r'}(M,\C)$ for some $r'>r$, which implies that also 
\begin{equation}
\label{sobolevnormspectral2}
\sum_{i\in \mathcal{I}}S_{r'}(\lambda_{i})
\left\vert\left<Y_{i},f\right>_{L^{2}} 
\right\vert^{2} = \Vert f\Vert_{H^{2}_{r'}}^{2}<\infty
\end{equation}
holds. Then, the sequence $i\mapsto a_{i}$ may be chosen to be an
arbitrary monotonically decreasing sequence,  for which the sum $\sum_{i\in
\mathcal{I}}S_{r'}(\lambda_{i})\left\vert a_{i}\right\vert^{2}$  is
divergent, as in that case, there always exists a threshold index,
above which  the relation (\ref{sobolevnormbound}) holds. Then,
$i\mapsto a_{i}$ may be  normalized in such a way that
(\ref{sobolevnormbound}) holds for any $i\in \mathcal{I}$.  Such a minimal multiplier would be
$\max\limits_{i\in\mathcal{I}}\frac{1}{a_{i}}\left\vert\left<
Y_{i},f\right>\right\vert$,  which, in practice, may always be
identified by the pertinent maximum on the stored finite orders. Note
that this approximation is exact whenever the threshold index is
reached within the stored orders.

Restricting again considerations to the case of a two-sphere, a suitable 
bounding sequence $(\ell,m)\mapsto a_{\ell}^{m}$ may be chosen as 
\begin{equation}
\label{s2sobolevbound}
(\ell,m)\mapsto
K\frac{1}{2\ell+1}\left(\frac{(\ell(\ell+1))^{r'+1}-1}{(\ell(\ell+1))-1}\right)^{-1/2},
\end{equation}
where $K$ is an unknown normalization factor. Then, the inequality 
\begin{eqnarray}
\label{tailsumbound}
&&\hskip-.8cm\sum_{\ell=\ell_{\mathrm{max}}+1}^{\infty}\sum_{m=-\ell}^{\ell}
\frac{(\ell(\ell+1))^{r+1}-1}{(\ell(\ell+1))-1}
\left\vert\left<Y{}_{\ell}^{m},f\right>_{L^{2}}\right\vert^{2}\cr 
&&\leq \vert K\vert^{2}\hskip-.3cm\sum_{\ell=\ell_{\mathrm{max}}+1}^{\infty}
\frac{(\ell(\ell+1))^{r+1}-1}{(\ell(\ell+1))^{r'+1}-1} \frac{1}{2\ell+1} 
\end{eqnarray}
holds for the $H^{2}_{r}$ norm-square of the tail sum, whenever $\vert K\vert$ is
chosen to be large enough such that the inequality
$\left\vert\left<Y{}_{\ell}^{m},f\right>_{L^{2}}\right\vert\leq \vert
K\vert\frac{1}{2\ell+1}\left(\frac{(\ell(\ell+1))^{r'+1}-1}{(\ell(\ell+1))-1}\right)^{-1/2}$
is satisfied for any allowed values of $\ell,m$, where {$r<r'$}. An immediate
lower bound for such $\vert K\vert$ may be given as 
\begin{eqnarray}
\label{normalization}
\max_{\ell=0,...,\infty} \max_{m=-\ell,...,\ell}
\sqrt{\frac{(\ell(\ell+1))^{r'+1}-1}{(\ell(\ell+1))-1}}
(2\ell+1)
\left\vert\left<Y{}_{\ell}^{m},f\right>_{L^{2}}\right\vert\,.  
\end{eqnarray}
By construction, the maximum value in question is necessarily attained
at certain finite
but unspecified $\ell$. In numerical applications, an estimate for $\vert K\vert$
may be obtained by assuming that the stored components up to the order
$\ell_{\mathrm{max}}$ already encode the relevant multipole orders for the field,
i.e., the multipole coefficients in the tail already fit into the
above scheme of the 
asymptotics. In this case, we may choose $K$ as 
\begin{eqnarray}
\label{normalizationapprox}
K=\hskip-.2cm \max_{\ell=0,...,\ell_{\mathrm{max}}} \max_{m=-\ell,...,\ell}
\sqrt{\frac{(\ell(\ell+1))^{r'+1}-1}{(\ell(\ell+1))-1}}
(2\ell+1)\left\vert\left<Y{}_{\ell}^{m},f\right>_{L^{2}}\right\vert.
\end{eqnarray}

It follows from the above observations that whenever $f\in C^{\infty}(\mathbb{S}^{2},\mathbb{C})$ then---since for any  $r^{'}$ $f\in H_{r^{'}}^{2}(\mathbb{S}^{2},\mathbb{C})$ also holds---its multipole coefficients $\left<Y{}_{\ell}^{m},f\right>_{L^{2}}$ shall decay faster than 
any polynomial in $\ell$. 
Note finally, that this upper bound on the $H^{2}_{r}$ norm of the tail sum error may be used, along with the Sobolev constant, to derive an upper
bound on the $C^{0}$ norm of the tail sum error.

\section{The Evaluation of the Integrals determining Radiation Anisotropy}
\label{radiationanisotropy}

In various physical applications it is important to determine the flux of conserved 
quantities through a hypersurface with topology $\mathbb{R}\times\mathcal{C}$.  

In the applied framework of multipole expansions the calculation of these type of quantities is rather straightforward, as, in general, they may be written as integrals of multilinear expressions of the basic variables over the compact surface $\mathcal{C}$. Thereby they can be given in terms of
$L^{2}(\mathcal{C},\mathbb{C})$  scalar products of some of the multilinear expressions of the basic field variables. Recall that multilinear expressions may be evaluated by purely spectral methods---as pointed out in \ref{sobolevembedding} and \ref{gauntcoefficients}---, whereas the $L^{2}$ scalar product may be calculated directly, using the fact that the eigenfunctions of the Laplace operator comprise an orthonormal system.

However, when we are interested in the scale of the anisotropy in the distribution of certain quantities, e.g.\,the energy radiated inwards or outwards, which cannot be expressed via the aforementioned integrals over the entire compact surface $\mathcal{C}$. Instead, the pertinent integrals have to be evaluated over a subset $\mathcal{A}\subset\mathcal{C}$ to compare the radiated flux through $\mathcal{A}$ and its complement in $\mathcal{C}$. These
type of integrals  may, however, be given as $\left<\cdot,\chi_{{}_{\mathcal{A}}} \cdot\right>_{L^{2}(\mathcal{C},\mathbb{C})}$ scalar products of multilinear expressions of basic field variables over $\mathcal{C}$, where $\chi_{{}_{\mathcal{A}}}$ denotes the characteristic function of the set $\mathcal{A}$.  As the multiplication by $\chi_{{}_{\mathcal{A}}}$ is a continuous map, all that we need for the spectral evaluation of this expression is the values of the $\left<Y_{i},\chi_{{}_{\mathcal{A}}}Y_{j}\right>_{L^{2}(\mathcal{C},\mathbb{C})}$ ($i,j\in\mathcal{I}$) matrix elements.

To demonstrate that this process is much more straightforward in practice than it sounds let us restrict our considerations again to the case of the
two-sphere. Then we have that 
\begin{eqnarray}
\label{partiall2s2}
&&\hskip-.4cm\left<Y{}_{\ell_{1}}^{m_{1}},\chi_{{}_\mathcal{A}}Y{}_{\ell_{2}}^{m_{2}}\right>_{L^2} = 
\int_{\mathcal{A}} \bar{Y}{}_{\ell_{1}}^{m_{1}} Y{}_{\ell_{2}}^{m_{2}} \cr&&\hskip-.9cm= (-1)^{m_{2}} 
\hskip-.4cm\sum_{k=0}^{\min(\ell1,\ell2)}\hskip-.4cm G_{\ell_{1},\ell_{2},\vert \ell_{1}-\ell_{2}\vert
  +2k}^{-m_{1},m_{2},-(m_{2}-m_{1})} \int_{\mathcal{A}} Y_{\vert
  \ell_{1}-\ell_{2}\vert+2k}^{m_{2}-m_{1}}, 
\end{eqnarray}
where the expansion of products of spherical harmonics was used, and the symbol 
$G$ stands for the Gaunt coefficients (\ref{gauntformula}). 

Let us further restrict considerations to an axially symmetric subset $\mathcal{A}\subset\mathcal{C}$. Then, in our parametrization, the integral $\int_{\mathcal{A}} Y_{\vert \ell_{1}-\ell_{2}\vert+2k}^{m_{2}-m_{1}}$ vanishes, unless $m_{2}=m_{1}$. Moreover, the integral 
$\int_{\mathcal{A}} Y_{\vert \ell_{1}-\ell_{2}\vert+2k}^{0}$ may be calculated by making use of 
the identity 
\begin{eqnarray}
P_{\ell+1}^{'}-P_{\ell-1}^{'}=(2\ell+1)P_{\ell}\,. 
\end{eqnarray}
This relation, which is valid for Legendre polynomials, may easily be verified by using the Rodriguez formula. As a consequence, the relation  
\begin{eqnarray}
\label{intyl0}
\int_{[0,\vartheta]\times[0,2\pi]} Y{}_{\ell}^{0} =  
\sqrt{\pi}\delta_{\ell,0} -\sqrt{\frac{\pi}{2\ell+1}} \left[P_{\ell+1}(\cos\vartheta) 
- P_{\ell-1}(\cos\vartheta)\right]
\end{eqnarray}
holds, where $\delta$ denotes Kronecker delta. Note that in evaluating the Legendre polynomials well-known computational methods can be applied (for a relevant implementation see, e.g.\,\cite{gsl}).


\end{document}